\documentclass[runningheads]{llncs}
\usepackage{etoolbox}
\AtBeginDocument{\pdfstringdefDisableCommands{\let\unskip=\relax}}
\usepackage[T1]{fontenc}

\usepackage{amsmath}
\usepackage{amssymb}
\usepackage{mathtools}          
\usepackage{bm}                 
\usepackage[mathscr]{eucal}     

\usepackage[linesnumbered,ruled,vlined]{algorithm2e}

\usepackage{graphicx}
\usepackage{xcolor}
\usepackage{tikz}
\usetikzlibrary{arrows.meta}
\usepackage{tabularx}
\usepackage{array}
\usepackage{subcaption}         
\usepackage{environ}            
\usepackage{float}

\usepackage{xspace}             

\usepackage{hyperref}
\usepackage{cleveref}

\spnewtheorem{observation}{Observation}{\bfseries}{\itshape}

\makeatletter
  \newcommand{\@problemtitle}{}
  \newcommand{\@probleminput}{}
  \newcommand{\@problemquestion}{}
  \newcommand{\problemtitle}[1]{\renewcommand{\@problemtitle}{#1}}
  \newcommand{\probleminput}[1]{\renewcommand{\@probleminput}{#1}}
  \newcommand{\problemquestion}[1]{\renewcommand{\@problemquestion}{#1}}

  \NewEnviron{probstat}{%
    \problemtitle{}\probleminput{}\problemquestion{}%
    \BODY
    \par\addvspace{.5\baselineskip}
    \noindent
    \begin{tabularx}{\textwidth}{@{\hspace{\parindent}} l X}
      \multicolumn{2}{@{\hspace{\parindent}}l}{\textbf{\@problemtitle}} \\
      \textbf{Input:} & \@probleminput \\
      \textbf{Goal:}  & \@problemquestion
    \end{tabularx}
    \par\addvspace{.5\baselineskip}
  }
\makeatother

\makeatletter
  \providecommand*{\cupdot}{%
    \mathbin{\mathpalette\@cupdot{}}%
  }
  \newcommand*{\@cupdot}[2]{%
    \ooalign{%
      $\m@th#1\cup$\cr
      \hidewidth$\m@th#1\cdot$\hidewidth
    }%
  }
\makeatother

\newcommand{\gs}{\mathsf{GS}}
\newcommand{\ws}{\mathsf{WS}}
\newcommand{\witN}{\mu_{WS}}

\newcommand{\vis}{\mathsf{Vis}}
\newcommand{\grN}{\mu_{GS}}

\newcommand{\bd}{\partial}

\newcommand{\wvp}{\mathsf{WV-Polygon}}
\newcommand{\WVP}{\mathscr{WVP}}

\newcommand{\AGP}{\mathsf{AGP}}

\newcommand{\WSP}{\mathsf{WSP}}
\newcommand{\dWSP}{\mathtt{DiscWSP}}
\newcommand{\conWSP}{\mathtt{ContWSP}}

\newcommand{\RA}{\mathsf{rAnchor}}
\newcommand{\LA}{\mathsf{\ell Anchor}}

\newcommand{\RC}{\mathsf{rWin}}
\newcommand{\LC}{\mathsf{\ell Win}}
\newcommand{\win}{\mathsf{win}}

\newcommand{\Rb}{\mathsf{rBdry}}
\newcommand{\Lb}{\mathsf{\ell Bdry}}

\newcommand{\I}{\mathtt{Int}}
\newcommand{\po}{\mathcal{P}}

\newcommand{\lb}{\mathscr{L}_{\mathtt{base}}}
\newcommand{\OO}{\mathcal{O}}

\newcommand{\eb}{{e}_{\mathtt{base}}}

\newcommand{\OPT}{\mathrm{OPT}}
\newcommand{\spr}{\mathsf{\Pi_{v}}}
\newcommand{\spl}{\mathsf{\Pi_{u}}}

\newcommand{\wv}{\mathcal{WV}}

\newcommand{\seg}[1]{\overline{#1}}        

\newcommand{\np}{{\sf NP}\xspace}
\newcommand{\xp}{{\sf XP}\xspace}

\newcommand{\nph}{{\sf NP}-hard\xspace}

\newcommand{\woh}{{\sf W}$[1]$-hard\xspace}

\crefname{theorem}{Theorem}{Theorems}       \Crefname{theorem}{Theorem}{Theorems}
\crefname{lemma}{Lemma}{Lemmas}             \Crefname{lemma}{Lemma}{Lemmas}
\crefname{claim}{Claim}{Claims}             \Crefname{claim}{Claim}{Claims}
\crefname{observation}{Observation}{Observations}
\Crefname{observation}{Observation}{Observations}
\crefname{definition}{Definition}{Definitions}
\Crefname{definition}{Definition}{Definitions}
\crefname{corollary}{Corollary}{Corollaries}
\Crefname{corollary}{Corollary}{Corollaries}
\crefname{proposition}{Proposition}{Propositions}
\Crefname{proposition}{Proposition}{Propositions}
\crefname{example}{Example}{Examples}       \Crefname{example}{Example}{Examples}
\crefname{note}{Note}{Notes}                \Crefname{note}{Note}{Notes}
\crefname{remark}{Remark}{Remarks}          \Crefname{remark}{Remark}{Remarks}
\crefname{figure}{Figure}{Figures}          \Crefname{figure}{Figure}{Figures}
\crefname{subfigure}{Figure}{Figures}       \Crefname{subfigure}{Figure}{Figures}
\crefname{section}{Section}{Sections}       \Crefname{section}{Section}{Sections}
\crefname{equation}{Equation}{Equations}    \Crefname{equation}{Equation}{Equations}
\crefname{algocf}{Algorithm}{Algorithms}    \Crefname{algocf}{Algorithm}{Algorithms}
\crefformat{theorem}{Theorem~#2#1#3}        \Crefformat{theorem}{Theorem~#2#1#3}
\crefformat{lemma}{Lemma~#2#1#3}            \Crefformat{lemma}{Lemma~#2#1#3}
\crefformat{claim}{Claim~#2#1#3}            \Crefformat{claim}{Claim~#2#1#3}
\crefformat{observation}{Observation~#2#1#3}
\Crefformat{observation}{Observation~#2#1#3}

\newcommand{\WVpath}{%
  (0,0) -- (0,4) -- (1,4) -- (0.6667,10.6667) -- (2,4) -- (5,4) -- (2.8,5.6) --
  (6,4) -- (9,4) -- (11.5,9) -- (10,4) -- (13,4) -- (16.2,5.6) -- (14,4) --
  (17,4) -- (18.6667,10.6667) -- (18,4) -- (21,4) -- (23,8) -- (22,4) --
  (24,4) -- (24,0) -- cycle}
 
\newcommand{\cone}[5][blue!45]{%
  \fill[#1,opacity=.40] (#2,#3) -- (#4,0) -- (#5,0) -- cycle;
  \draw[line width=.3pt,black!65] (#2,#3) -- (#4,0) (#2,#3) -- (#5,0);}
 
\newcommand{\vlab}[3][2.2pt]{\fill (#2,4) circle (1.5pt);
  \node[below=#1,scale=.60,fill=white,fill opacity=.8,text opacity=1,
        inner sep=.5pt] at (#2,4) {$v_{#3}$};}
 
\newcommand{\Hgt}{4}
\newcommand{\tA}{2.2857}\newcommand{\tB}{4.5714}   
\newcommand{\tC}{5.7143}\newcommand{\tD}{8.0000}   
\newcommand{\tE}{9.1429}\newcommand{\tF}{11.4286}  
\newcommand{\tG}{12.5714}\newcommand{\tH}{14.8571} 
\newcommand{\tI}{16.0000}\newcommand{\tJ}{18.2857} 
\newcommand{\tK}{19.4286}\newcommand{\tL}{21.7143} 
 
\newcommand{\trap}[5][blue!45]{%
  \fill[#1,opacity=.55] (#2,\Hgt) -- (#3,\Hgt) -- (#5,0) -- (#4,0) -- cycle;
  \draw[line width=.4pt] (#2,\Hgt) -- (#3,\Hgt) -- (#5,0) -- (#4,0) -- cycle;}
\newcommand{\traptd}[5][black!45]{%
  \fill[#1,opacity=.30] (#2,\Hgt) -- (#3,\Hgt) -- (#5,0) -- (#4,0) -- cycle;
  \draw[line width=.6pt,densely dotted]
        (#2,\Hgt) -- (#3,\Hgt) -- (#5,0) -- (#4,0) -- cycle;}
 
\newcommand{\rails}{%
  \path (-1.9,-1.8) rectangle (25.6,6.6);
  \draw[line width=1pt] (-0.5,\Hgt) -- (24.6,\Hgt);
  \draw[line width=1pt] (-0.5,0)    -- (24.6,0);
  \node[left] at (-0.6,\Hgt) {$L_t$};
  \node[left] at (-0.6,0)    {$L_b$};}
 
\newcommand{\topticks}{%
  \foreach \x in {0,1.1429,2.2857,3.4286,4.5714,5.7143,6.8571,8.0,9.1429,
                  10.2857,11.4286,12.5714,13.7143,14.8571,16.0,17.1429,
                  18.2857,19.4286,20.5714,21.7143,22.8571,24.0}
     {\fill (\x,\Hgt) circle (1.5pt);}
  \foreach \x/\lb in {0/0,\tA/2,\tB/4,\tC/5,\tD/7,\tE/8,\tF/{10},\tG/{11},
                      \tH/{13},\tI/{14},\tJ/{16},\tK/{17},\tL/{19},24/{21}}
     {\node[above=0pt,scale=.60] at (\x,\Hgt) {$v_{\lb}$};}}
 
\newcommand{\bt}[3][1pt]{\fill (#2,0) circle (1.7pt);
  \node[below=#1,scale=.70] at (#2,0) {#3};}
 
\newcommand{\frontier}{%
  \draw[line width=1.3pt,densely dashed,red!75!black] (8.8,0) -- (\tF,\Hgt);
  \fill[red!75!black] (8.8,0) circle (2.3pt);
  \fill[red!75!black] (\tF,\Hgt) circle (2.3pt);}
\newcommand{\rightzone}{%
  \fill[black!7] (\tF,\Hgt) -- (24.6,\Hgt) -- (24.6,0) -- (8.8,0) -- cycle;}

\newif\ifauthorcomments
\authorcommentsfalse   

\providecommand{\marginnote}[2][]{}

\newcounter{shouvik}

\newcounter{sasanka}

\newcounter{udvas}

\makeatletter
\def\@fnsymbol#1{\ensuremath{\ifcase#1\or \bullet\or \dagger\or \ddagger\or
  \mathsection\or \mathparagraph\or \|\or {\dagger\dagger}\or
  {\ddagger\ddagger}\else\@ctrerr\fi}}
\makeatother

\usepackage[most]{tcolorbox}
\newtcolorbox{probbox}[1]{%
  enhanced, breakable,
  colback=white, colframe=black!70, boxrule=0.5pt, arc=1.2mm,
  left=3mm, right=3mm, top=2mm, bottom=2mm,
  attach boxed title to top left={xshift=4mm, yshift=-\tcboxedtitleheight/2},
  boxed title style={colback=white, colframe=white, boxrule=0pt,
                     sharp corners, left=1mm, right=1mm, top=0mm, bottom=0mm},
  coltitle=black, fonttitle=\normalsize, title={#1}}

\title{Witness Set in Weak Visibility Polygons is Polynomial-Time Solvable}

\titlerunning{Witness Set in Weak Visibility Polygons}
\authorrunning{U. Das, S. Mondal, S. Roy}

\ifdefined\iflatexml\else\expandafter\newif\csname iflatexml\endcsname\fi

\iflatexml
\author{Udvas Das\inst{1}
\and Shouvik Mondal\inst{2}
\and Sasanka Roy\inst{3}}
\institute{Advanced Computing and Microelectronics Unit (ACMU),\\
Indian Statistical Institute, Kolkata, India\\
\email{udvas.das@gmail.com}
\and Advanced Computing and Microelectronics Unit (ACMU),\\
Indian Statistical Institute, Kolkata, India\\
\email{shouvik.math@gmail.com}
\and Advanced Computing and Microelectronics Unit (ACMU),\\
Indian Statistical Institute, Kolkata, India\\
\email{sasanka.ro@gmail.com}}
\else
\author{Udvas Das\inst{1}
\and Shouvik Mondal\inst{1}
\and Sasanka Roy\inst{1}}
\institute{Advanced Computing and Microelectronics Unit (ACMU),\\
Indian Statistical Institute, Kolkata, India\\
\email{udvas.das@gmail.com, shouvik.math@gmail.com, sasanka.ro@gmail.com}}
\fi

\begin{document}
\maketitle

\begin{abstract}

In the classical {\sc Art Gallery Problem} $(\AGP)$, guards are placed in a polygon to see every point together~\cite{orourke1987artgallery}. The {\sc Witness Set Problem} $(\WSP)$, introduced by the authors in~\cite{DBLP:journals/ijcga/AmitMP10}, is a natural dual to the $\AGP$. In this paper, we study the {\sc Witness Set Problem} in weak visibility polygons $(\WVP)$, the simple polygons in which every point is seen from one fixed edge.
A {witness set} is a set of points whose visibility regions are pairwise disjoint, so that no single guard sees two of them. A maximum witness set, therefore, lower-bounds the guard number. Exact polynomial-time algorithms for $\WSP$ are known only for monotone mountains~\cite{DAESCU201922}, which is a proper subclass of $\WVP$s. We give the first exact polynomial-time algorithms for $\WSP$ in weak visibility polygons in two settings.\\
\begin{description}
    \item {\sc Discrete Witness Set Problem} $(\dWSP)$: Here, the witnesses come from a given set of $m$ points, and we find a maximum witness subset in $\OO(n + m \log(n + m))$ time on an $n$-vertex polygon. Crucially, the visibility intersection graph inside $\WVP$s, where two points are adjacent if their visibility regions intersect, properly contains interval graphs and permutation graphs~\cite{GOLUMBIC20041}. Moreover, this graph class turns out to be a subclass of trapezoid graphs, the intersection
    graph of trapezoids between two parallel lines~\cite{DaganGP88,FelsnerMW97}. We also
    give a matching lower bound, so our algorithm for $\dWSP$ problem is optimum. The characterization of graphs could find interesting applications in graph theory.
    \vspace{1mm}
    \item {\sc Continuous Witness Set Problem} $(\conWSP)$: Here, the witnesses are any point of the polygon, and we give an exact polynomial-time algorithm running in $\OO(n \log n + \rho^{2}(n + \rho^{2}))$ time, with $\rho$ reflex vertices.
\end{description}

\keywords{Witness set \and Weak visibility polygon \and Art gallery
problem \and Trapezoid graph \and Exact algorithm}
\end{abstract}

\raggedbottom

\section{Introduction}
\label{F-sec:intro}

Picture a long gallery with a glass wall on one side, a strait watched from
one shore, or a valley overlooked by a road. In each case, the region to be
watched is large and two-dimensional, but the cameras, guards, or sensors may
sit only along one fixed edge of it. The geometric model for such a region is
a \emph{weak visibility polygon}, a simple polygon in which every point is
seen from some point of one fixed edge, called the base edge.

The general form of this question is the classical \textsc{Art Gallery
Problem}. Victor Klee posed it in 1973 and asked how many guards are needed to see all of a gallery with $n$ walls. Chv\'atal proved that
$\lfloor n/3 \rfloor$ guards are always sufficient and sometimes
necessary~\cite{DBLP:journals/jctb/Chvatal75}, and Fisk later gave a short
proof by triangulation and three
colouring~\cite{Fisk78}. The algorithmic version is much harder. Lee and Lin
proved that computing a minimum guard set is
\nph~\cite{DBLP:journals/tit/LeeL86}. The authors in~\cite{EidenbenzSW01}
proved APX-hardness for point, vertex, and edge guards, and the authors
in~\cite{DBLP:journals/jacm/AbrahamsenAM22} proved that the problem is
$\exists \mathbb{R}$-complete, so it is not even known to lie in \np. Bonnet
and Miltzow proved \woh ness for the number of
guards~\cite{BonnetMiltzow20}. On the positive side, Ghosh gave an
$\OO(\log n)$ approximation algorithm for vertex and edge
guards~\cite{GHOSH2010718}, King and Kirkpatrick improved the ratio to
$\OO(\log \log \OPT)$~\cite{KingKirkpatrick11}, and an exact algorithm running in
$n^{\OO(k)}$ time follows from tools of real algebraic geometry, an observation
of Sharir reported by Efrat and Har-Peled~\cite{EfratHarPeled06}.

\paragraph{\textbf{Weak Visibility Polygons.}}
This class is more than a convenient special case. The known constant factor
approximation algorithms for simple polygons are built by cutting the polygon
into a hierarchy of weak visibility polygons. It is also hard in its own
right. The authors in~\cite{BHATTACHARYA2017109} proved that point guarding a
weak visibility polygon is \nph. In the same work, they gave a
$6$-approximation algorithm for vertex guarding such a polygon without holes,
running in $\OO(n^2)$ time, and they showed that with holes, no polynomial time
algorithm achieves a ratio better than $((1-\epsilon)/12)\ln n$ unless
$\np = \mathsf{P}$. Katz gave a PTAS for vertex guarding the vertices of a
weak visibility polygon and for vertex guarding its
boundary~\cite{DBLP:journals/corr/abs-1803-02160}. The authors
in~\cite{DBLP:journals/jocg/AshurFK21} gave a constant factor approximation
algorithm for vertex guarding such a polygon, and the authors
in~\cite{DBLP:journals/comgeo/AshurFKS22} introduced terrain-like graphs,
which yield further PTASs for guarding it from its vertices.

\paragraph{\textbf{The Witness Set Problem.}}
All the results above bound the guard number from above. A practitioner also
needs a certificate from below. The authors
in~\cite{DBLP:journals/ijcga/AmitMP10} supplied one. They placed
\emph{visibility-independent} points, which they called witnesses, and used
them to certify lower bounds on the guard number. A \emph{witness set} is a
set of points whose visibility regions are pairwise disjoint. No single guard
sees two witnesses, so the size of a maximum witness set is a lower bound on
the number of guards. This is the \textsc{Witness Set Problem}. A related
packing notion is a \emph{hidden set}, a set of points no two of which see
each other; the authors in~\cite{BrowneKMP23} gave the first constant factor
approximation algorithms for it. A witness set is a stronger requirement,
since it asks that the visibility regions be disjoint and not merely that the
points fail to see one another.

The witness number is interesting for a second reason. Call a polygon
\emph{perfect} when its guard number equals its witness number. For a perfect
polygon, the lower bound is tight, so computing a maximum witness set also
solves the art gallery problem. Guarding a $1.5$D terrain, the closest
relative of a weak visibility polygon, is \nph, as King and Krohn
proved~\cite{kingkrohn2011terrain}. Yet the authors in~\cite{DAESCU201922}
gave an optimal linear time algorithm for guarding a terrain by guards on a
horizontal line, and proved that the minimum guard number equals the
maximum witness number. The same results apply to monotone mountains, which
are $x$-monotone polygons whose lower chain is a single edge. Every monotone
mountain is therefore perfect, and this is the first non-trivial class of perfect polygons.

Two recent works widen the picture. The authors
in~\cite{DBLP:journals/corr/abs-2511-10224} gave a polynomial time algorithm
for the discrete version in general simple polygons, and for monotone polygons
with $r$ reflex vertices, an exact algorithm running in
$r^{\OO(k)} \cdot n^{\OO(1)}$ time for witness sets of size $k$, together with
a PTAS. The authors in~\cite{DBLP:journals/corr/abs-2605-01592} proved that the witness set problem in simple polygons lies in $\np \cap \xp$. This is a
sharp contrast with the art gallery problem, which is not in \np unless
$\np = \exists \mathbb{R}$.

Monotone mountains are a proper subclass of weak visibility polygons. Outside
that subclass, no exact polynomial-time algorithm for witness sets was known.
This paper closes the gap.

\subsection{Our Contributions}
\label{F-subsec:contributions}

The Problem Statements are defined as follows.

\begin{probbox}{\textsc{Discrete Witness Set} $(\dWSP)$ in $\wvp$}
\textbf{Input:}\ \ A weak visibility polygon $\wv$ with base edge $\eb$, and a
finite set $S = \{s_1,\dots,s_m\}$ of points of $\wv$.\\
\textbf{Task:}\ \ Find a maximum-cardinality subset $Q \subseteq S$ such that
$\vis(w) \cap \vis(w') = \emptyset$ for all distinct $w, w' \in Q$, that is,
a witness set $\ws(\wv,S)$.
\end{probbox}

\begin{probbox}{\textsc{Continuous Witness Set} $(\conWSP)$ in $\wvp$}
\textbf{Input:}\ \ A weak visibility polygon $\wv$ with base edge $\eb$.\\
\textbf{Task:}\ \ Find a maximum-cardinality set $W \subseteq \wv$ such that
$\vis(w) \cap \vis(w') = \emptyset$ for all distinct $w, w' \in W$, that is,
a witness set $\ws(\wv,\wv)$.
\end{probbox}

We give the first exact polynomial time algorithms for witness sets in weak
visibility polygons.\footnote{In a companion manuscript, we study the dual
question for the same class, namely placing the fewest guards on the base edge
so that together they see the whole polygon~\cite{sgp}.} In $\dWSP$, the
witnesses come from a given set $S$ of $m$ points, and for this problem, our
algorithm is near-linear and optimal. In $\conWSP$, a witness may be any point
of the polygon. Throughout, $\wv$ is a weak visibility polygon with $n$ vertices and $\rho$ reflex vertices. Our first result is structural.

\begin{theorem}[\cref{F-thm:trapezoid-graph}]
The visibility intersection graph $G_S$ of an instance $(\wv, S)$ is the
intersection graph of a family of trapezoids between two parallel lines. In
particular, $G_S$ is a trapezoid graph, and a trapezoid representation of it is
read off from the anchors of the points of $S$.
\end{theorem}

\begin{theorem}[\cref{F-thm:disws-runtime} and \cref{F-thm:disws-lower}]
$\dWSP$ in a weak visibility polygon is solvable in $\OO(n + m \log (n+m))$
time and $\OO(n+m)$ space. In the algebraic decision-tree model, it requires $\Omega(k \log k)$ operations in the worst case on instances with $k$ points
in a polygon of $\OO(k)$ vertices. For $m = \Theta(n)$ the problem is
therefore solved in $\Theta(n \log n)$ time, and no algebraic decision-tree algorithm does better.
\end{theorem}

\begin{theorem}[\cref{F-thm:VIGgraphclasslocation}]
Let $\mathsf{VIG}$ be the class of visibility intersection graphs of weak
visibility polygons, and let $\mathsf{Int}$, $\mathsf{Perm}$ and
$\mathsf{Trap}$ be the classes of interval, permutation, and trapezoid graphs.
Then $\mathsf{Int} \subsetneq \mathsf{VIG}$,
$\mathsf{Perm} \subsetneq \mathsf{VIG}$ and
$\mathsf{VIG} \subseteq \mathsf{Trap}$.
\end{theorem}

\begin{theorem}[\cref{F-thm:ws-correct} and \cref{F-thm:ws-time}]
$\conWSP$ in a weak visibility polygon is solvable exactly in
$\OO\!\left(n \log n + \rho^{2}(n + \rho^{2})\right)$ time.
\end{theorem}

\paragraph{\textbf{Applications.}}
Witness sets are not only a proof device. Practical solvers for the art
gallery problem are primal-dual or integer-programming methods, and they need
a lower bound at every iteration to certify that a guard set is
optimal. A witness set is exactly such a certificate, and the authors
in~\cite{DBLP:journals/ijcga/AmitMP10} computed witnesses for this purpose.
Exact solvers built on the same idea appear
in~\cite{KrollerBFS12,TozoniRS16}. The underlying placement question is itself
applied. The authors in~\cite{BenMosheKM07} motivate terrain guarding by the
placement of antennas for line-of-sight communication and of street lights and
cameras along roads. Gonz\'alez-Banos and Latombe used an art gallery
formulation for sensor placement in robotics~\cite{GonzalezBanosLatombe01},
and coverage of this kind is standard in sensor network
design~\cite{Wang11}. A weak visibility polygon models exactly the case where
the sensors must sit along one fixed wall, shore, or road.

\section{Preliminaries}
\label{F-sec:prelims}

We collect the definitions and notation used throughout. We begin with
simple polygons and weak visibility polygons, then define witness sets
and guard sets together with their cardinalities, and finally recall the
anchors and the interval a point induces on the base edge.

\begin{definition}[Simple Polygon]
\label{F-def:simple-polygon}
A \emph{simple polygon} $\po$ is the closed region bounded by a finite
closed chain of straight, pairwise non-crossing edges that meet only at
their common endpoints, the \emph{vertices}. We write $V(\po)$ for its
vertex set and $\bd(\po)$ for its boundary. A vertex is \emph{reflex}
if the interior angle of $\po$ at it exceeds $\pi$, and \emph{convex}
otherwise. Two points $a$ and $b$ of $\po$ \emph{see} each other if
the segment $\seg{ab}$ lies in $\po$. We write $\vis(p)$ for the set of
points of $\po$ that $p$ sees.
\end{definition}

For a segment, its \emph{relative interior} is the segment without its
two endpoints. A \emph{chord} of a simple polygon $\po$ is a segment
whose two endpoints lie on $\bd(\po)$ and whose relative interior lies
in the interior of $\po$. A chord splits $\po$ into two subpolygons,
and it splits $\bd(\po)$ into two \emph{arcs}, one on each side.

\begin{definition}[Weak Visibility Polygon]
\label{F-def:wvp}
A simple polygon $\wv$ is a \emph{weak visibility polygon} $(\wvp)$
with respect to one of its edges $\eb = uv$, the \emph{base edge}, if every point of $\wv$ is visible from some point of
$\eb$~\cite{ghosh2007vis}. We take the vertices of $\wv$ to be
$v_0, v_1, \dots, v_n$ in clockwise order with $\eb = v_0 v_n$, and we
write $u = v_0$ and $v = v_n$ for the two endpoints of the base. The polygon
has $n+1$ vertices; we suppress this additive constant and call $n$ its size.
\end{definition}

A \emph{base point} is a point of the base edge $\eb$.

\begin{definition}[Witness Set]
\label{F-def:witness-set}
A set $W \subseteq \wv$ is a \emph{witness
set}~\cite{DBLP:journals/ijcga/AmitMP10} if no point of $\wv$ sees two
of its members, that is, if $\vis(w) \cap \vis(w') = \emptyset$ for
every two distinct $w, w' \in W$. We write $\ws(\wv, Q)$ for a witness
set in the polygon $\wv$ with the witness points drawn from a set $Q \subseteq \wv$, and $\witN(\wv, Q)$ for the cardinality of such a set.
\end{definition}

\begin{definition}[Guard Set]
\label{F-def:guard-set}
A set $G \subseteq \wv$ is a \emph{guard set}~\cite{Belleville91} if
every point of $\wv$ is visible from some $g \in G$. We write
$\gs(\wv, A)$ for a guard set in $\wv$ with the guards drawn from $A \subseteq \wv$, and $\grN(\wv, A)$ for the cardinality of such a set.
\end{definition}

A witness set and a guard set bound each other: one guard sees at most
one witness, so $\witN(\wv, Q) \le \grN(\wv, A)$ whenever $Q$ and $A$
range over all of $\wv$, which is why a maximum witness set certifies a
guard lower bound.

For a point $p \in \wv$, let $\spl(p)$ and $\spr(p)$ denote the
Euclidean shortest paths inside $\wv$ from $p$ to $u$ and to $v$,
respectively. The next two notions, from
\cite{DBLP:journals/corr/abs-2511-10224}, turn visibility on the base
into one-dimensional data.

\begin{definition}[Anchors]
\label{F-def:prelim-anchors}
The \emph{left anchor} $\LA(p)$ of a point $p \in \wv$ is the first
reflex vertex on $\spl(p)$, and the \emph{right anchor} $\RA(p)$ is the
first reflex vertex on $\spr(p)$. If the shortest path is a straight
segment, the anchor is the base endpoint it reaches.
\end{definition}

\begin{definition}[Window]
\label{F-def:prelim-window}
Let $x$ be a point of $\wv$ and let $\mu$ be a vertex of $\wv$ that $x$
sees and at which the ray from $x$ through $\mu$ enters the interior of
$\wv$. The \emph{window of $x$ at $\mu$} is the segment
$\seg{\mu\,\win(x,\mu)}$, where $\win(x,\mu)$ denotes the first point of
$\bd(\wv)$ that this ray meets after $\mu$. When $\mu$ is an endpoint of
$\eb$ we set $\win(x,\mu) = \mu$, so that the window is the single point
$\mu$.
\end{definition}

\begin{observation}
\label{F-obs:window}
A window of $x$ at a reflex vertex $\mu$ is a chord of $\wv$, it is the
edge of $\vis(x)$ at which $\mu$ blocks the view of $x$, and the segment
$\seg{x\,\win(x,\mu)}$ lies in $\vis(x)$.
\end{observation}

\begin{definition}[Interval on the Base]
\label{F-def:prelim-interval}
The \emph{base interval} of a point $p \in \wv$ is
$\I(p) = \vis(p) \cap \eb$. It is a single nonempty closed segment
$[\ell(p), r(p)] \subseteq
\eb$~\cite{DBLP:journals/corr/abs-2511-10224}, so a base point
$g \in \eb$ sees $p$ if and only if $g \in \I(p)$.
\end{definition}

The two endpoints of the base interval are window endpoints, and we record this together with the two segments we use below.

\begin{observation}
\label{F-obs:anchor-collinear}
For every $p \in \wv$ the three points $p$, $\LA(p)$, $\ell(p)$ are
collinear and occur in this order, and so do $p$, $\RA(p)$, $r(p)$;
hence $\ell(p) = \win(p, \LA(p))$ and $r(p) = \win(p, \RA(p))$. The
\emph{boundary segments} of $p$ are $\Lb(p) = \seg{p\,\ell(p)}$ and
$\Rb(p) = \seg{p\,r(p)}$; both lie in $\vis(p)$.
\end{observation}

Throughout the paper, we assume \emph{general position}: no three vertices of
$\wv$ are collinear.

\section{\textsc{Discrete Witness Set Problem} in $\wvp$}
\label{F-sec:disws-wvp}

Restricting the witnesses to a finite input set turns the witness
problem into a maximum independent set problem on a trapezoid graph,
and the conflict relation among the input points is then decided by four numbers per point. We prove this correspondence and derive a
near-linear-time exact algorithm from it; when the point set is
proportional to the polygon the running time is $\Theta(n \log n)$,
which we show is optimal in the algebraic decision-tree model.

Throughout this section, $\wv$ is a simple polygon with vertices
$v_0, v_1, \dots, v_n$ in clockwise order that is weakly visible
from its base edge $\eb = v_0 v_n$. The vertex $v_0$ is the left
endpoint of $\eb$, and the vertex $v_n$ is the right endpoint, so
$v_0, v_1, \dots, v_n$ is the chain of $\wv$ other than $\eb$. We
call this chain the \emph{upper chain}. For two points $x$ and $y$
of the upper chain we write $x \prec y$ when $x$ precedes $y$ in the
clockwise order. We orient $\eb$ from $v_0$ to $v_n$, and for two
points $x$ and $y$ of $\eb$ we write $x < y$ when $x$ precedes $y$
in this orientation. The words \emph{left} and \emph{right} always
refer to these two orders. We also write $u = v_0$ and $v = v_n$ for
the two endpoints of $\eb$, and $\spl(p)$, $\spr(p)$ for the shortest
paths of \cref{F-sec:prelims} from $p$ to $u$ and to $v$.

\begin{definition}[Discrete Witness Set]
\label{F-def:disws}
An instance of \textsc{Discrete Witness Set} consists of a weak
visibility polygon $\wv$ and a finite set $S = \{s_1, \dots, s_m\}$ of
points of $\wv$. A subset $Q \subseteq S$ is a \emph{witness set} if
$\vis(w) \cap \vis(w') = \emptyset$ for every two distinct points $w$
and $w'$ of $Q$. The problem asks for a witness set of maximum
cardinality. This is a witness set $\ws(\wv, S)$ in the notation of
\cref{F-sec:prelims}, and the problem asks for its size $\witN(\wv, S)$.
\end{definition}

In addition to the general position assumed in \cref{F-sec:prelims}, no point of $S$ lies on a line through two vertices of $\wv$.

\begin{figure}[ht!]
    \centering
    \includegraphics[scale=0.9]{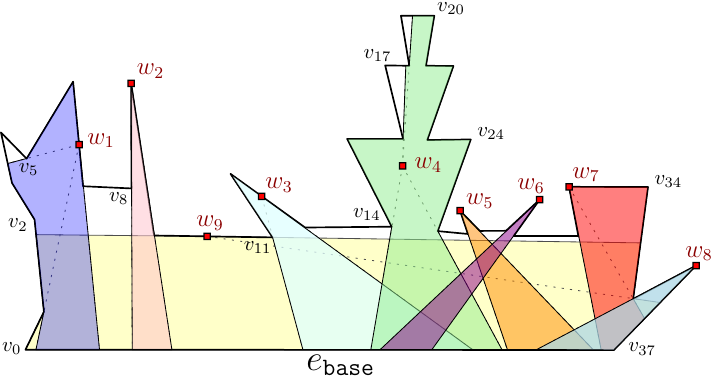}
    \caption{The \textsc{Discrete Witness Set} problem. Here
    $S = \{w_1, \dots, w_9\}$, and not all points of $S$ lie on
    $\bd(\wv)$. The subset $Q = \{w_1, w_2, w_4, w_5, w_7\}$ is a valid
    witness set in $\wv$, and its visibility regions are pairwise
    disjoint, so $\witN(\wv, S) = 5$.}
    \label{F-fig:WVPWitness}
\end{figure}


The witness sets of an instance are the independent sets of one
graph, which we now name.

\begin{definition}[Visibility Intersection Graph]
\label{F-def:conflict-graph}
The \emph{visibility intersection graph} (VIG) $G_S$ of an instance $(\wv, S)$ has vertex
set $S$, and two distinct points $s$ and $s'$ of $S$ are adjacent if
and only if $\vis(s) \cap \vis(s') \neq \emptyset$.
\end{definition}

A subset of $S$ is a witness set if and only if it is an independent
set of $G_S$. The whole difficulty, therefore, lies in the structure
of $G_S$. We show in \cref{F-subsec:disws-correctness} that $G_S$ is
a trapezoid graph, and that a trapezoid representation of it is read
off from the anchors of the points of $S$.

\subsection{Windows and Base Intervals}
\label{F-subsec:disws-anchors}

We use the anchors $\LA(p)$ and $\RA(p)$ and the base interval
$\I(p) = [\ell(p), r(p)]$ of \cref{F-sec:prelims}, together with the
boundary segments $\Lb(p) = \seg{p\,\ell(p)}$ and
$\Rb(p) = \seg{p\,r(p)}$. By \cref{F-obs:anchor-collinear} the point $p$
has a window at each of its two anchors, and we write
$\LC(p) = \seg{\LA(p)\,\ell(p)}$ and $\RC(p) = \seg{\RA(p)\,r(p)}$ for
them, see \cref{F-fig:anchor}. Both lie in $\vis(p)$, and by \cref{F-obs:window} each is a chord of
$\wv$ unless its anchor is an endpoint of $\eb$, in which case it is the
single point $v_0$ or $v_n$. These two are the leftmost and the
rightmost windows of $\vis(p)$, and they are the only ones the analysis
uses.


\begin{figure}[H]
    \centering
    \includegraphics[width=0.9\linewidth]{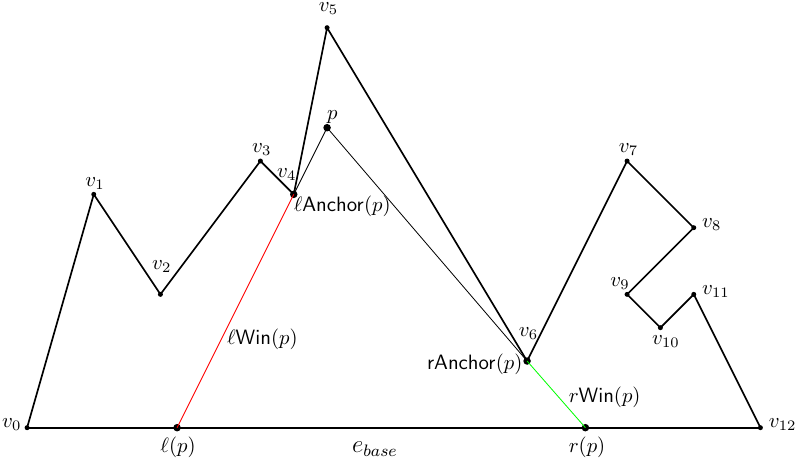}
    \caption{Illustrative figure for Anchors and Windows (Definitions~\ref{F-def:prelim-anchors},~\ref{F-def:prelim-window})}
    \label{F-fig:anchor}
\end{figure}

General position gives $\ell(p) < r(p)$ for every point $p$ of $S$.
Indeed, an equality would place $\LA(p)$ and $\RA(p)$ on a
common line through $p$, and these are two vertices of $\wv$.

By \cref{F-obs:window} the right window of $p$ is a chord of $p$ is a chord of
$\wv$ whenever it is not a single point. It therefore splits $\wv$
into two closed subpolygons. We write $\wv^-_r(p)$ for the one
whose boundary arc contains $v_0$ and $\wv^+_r(p)$ for the one whose
boundary arc contains $v_n$, and we set $\wv^-_r(p) = \wv$ when
$r(p) = v_n$. Symmetrically, $\LC(p)$ splits $\wv$ into
$\wv^-_l(p)$ and $\wv^+_l(p)$, and we set $\wv^+_l(p) = \wv$ when
$\ell(p) = v_0$.

We record one fact about chords that the proofs below reuse.

\begin{observation}
\label{F-obs:chord-cross}
A chord of $\wv$ splits $\wv$ into two subpolygons. A second chord
crosses it if and only if the two endpoints of the second chord lie on
the two different boundary arcs. Equivalently, two chords cross if and
only if their four endpoints interleave in the boundary order.
\end{observation}

\subsection{The Trapezoid Representation}
\label{F-subsec:disws-construction}

The construction uses two horizontal parallel lines. Let $L_b$ be
the lower line and let $L_t$ be the upper line. Let
$\beta \colon \eb \to L_b$ be an isometry that preserves the left to
right order of $\eb$, and let $\tau$ map the vertex $v_i$ to the
point of abscissa $i$ on $L_t$, for $i = 0, 1, \dots, n$. The map
$\tau$ preserves the clockwise order of the upper chain by
construction.

\begin{definition}[Trapezoid of a point]
\label{F-def:trapezoid}
The \emph{trapezoid} $T(p)$ of a point $p \in S$ is the convex hull
of the upper side $\seg{\tau(\LA(p))\,\tau(\RA(p))}$ and
the lower side $\seg{\beta(\ell(p))\,\beta(r(p))}$.
\end{definition}


\begin{figure}[ht!]
    \centering
    \includegraphics[
    width=1.0\textwidth,
    height=0.21\textheight
]{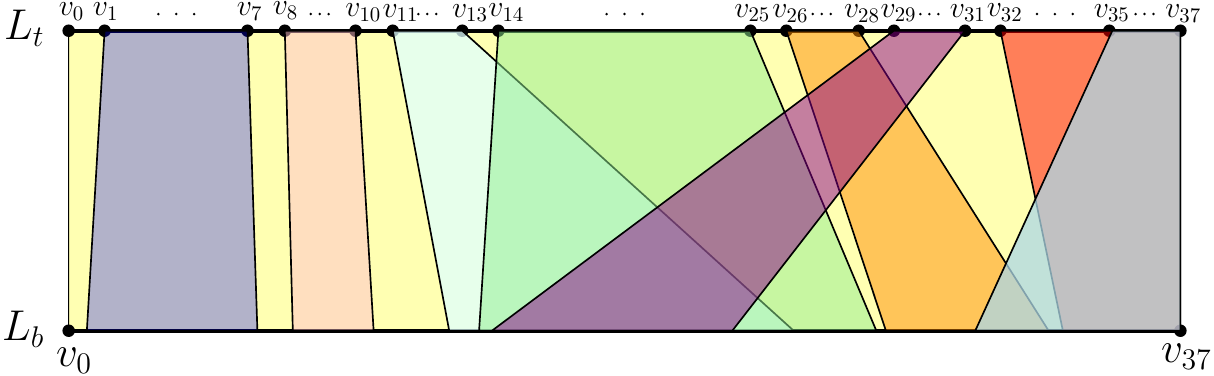}
    \caption{The trapezoid representation of \cref{F-def:trapezoid} for the
    instance of \cref{F-fig:WVPWitness}. The base edge is drawn as the lower
    line $L_b$ and the upper chain as the upper line $L_t$; by
    \cref{F-thm:trapezoid-graph} two points of $S$ are witnesses exactly when their trapezoids are disjoint.}
    \label{F-fig:WVPTrap}
\end{figure}

\Cref{F-fig:WVPTrap} shows the representation built from the instance of
\cref{F-fig:WVPWitness}.

\Cref{F-def:trapezoid} is legitimate only if the upper side has its
left endpoint before its right endpoint. The following lemma
supplies this, and the inequality $\ell(p) < r(p)$ supplies the
same for the lower side.

\begin{lemma}
\label{F-lem:anchor-order}
Every point $p \in S$ satisfies $\LA(p) \prec \RA(p)$.
\end{lemma}

\begin{proof}
If $\LA(p) = v_0$, then the claim holds because $\RA(p)$
is a reflex vertex or is $v_n$, and no such vertex is $v_0$. The
case $\RA(p) = v_n$ is symmetric. So assume that both windows
are chords. The windows are contained in $\Lb(p)$ and in
$\Rb(p)$ respectively, these two segments meet only at $p$, and
neither window contains $p$. Hence, the two windows are disjoint,
and their four endpoints are distinct. Read $\bd(\wv)$ clockwise
from $v_0$. The upper chain comes first in the order
$v_0, v_1, \dots, v_n$, and the points of $\eb$ come afterward from
right to left. Suppose that $\RA(p) \prec \LA(p)$. Since
$\ell(p) < r(p)$, the four endpoints then appear in the clockwise
order $\RA(p)$, $\LA(p)$, $r(p)$, $\ell(p)$. The four
endpoints interleave in this order, so the two windows cross by
\cref{F-obs:chord-cross}. This contradicts their disjointness.
\end{proof}

\subsection{The Algorithm}
\label{F-subsec:disws-algorithm}

Here, we present the pseudocode of the algorithm. 

\begin{algorithm}[htbp]
\caption{Maximum witness subset of a finite point set in a weak
 visibility polygon.}
\label{F-alg:disws-wvp}
\KwIn{A weak visibility polygon $\wv$ with vertices
  $v_0, \dots, v_n$ in clockwise order and base edge
  $\eb = v_0 v_n$, and a set $S = \{s_1, \dots, s_m\}$ of points of
  $\wv$.}
\KwOut{A maximum witness subset $Q \subseteq S$.}
\BlankLine
Triangulate $\wv$\;
Build the shortest path map of $\wv$ with source $v_0$ and the
shortest path map with source $v_n$, each with a point location
structure\;
\ForEach{$s \in S$}{
  $\LA(s) \gets$ the second point of $\spl(s)$, read from
  the map with source $v_0$\;
  $\RA(s) \gets$ the second point of $\spr(s)$, read from
  the map with source $v_n$\;
  \lIf{$\LA(s) = v_0$}{$\ell(s) \gets v_0$}
  \lElse{$\ell(s) \gets$ the point in which the ray from $s$ through
    $\LA(s)$ meets $\eb$}
  \lIf{$\RA(s) = v_n$}{$r(s) \gets v_n$}
  \lElse{$r(s) \gets$ the point in which the ray from $s$ through
    $\RA(s)$ meets $\eb$}
}
Fix the two parallel lines $L_b$ and $L_t$, the isometry
$\beta \colon \eb \to L_b$, and the map $\tau$ of
\cref{F-subsec:disws-construction}\;
\lForEach{$s \in S$}{$T(s) \gets$ the trapezoid of
  \cref{F-def:trapezoid}}
$Q \gets$ a maximum independent set of the trapezoid graph
represented by $\{T(s) : s \in S\}$, computed by the algorithm
of~\cite{FelsnerMW97}\;
\Return $Q$\;
\end{algorithm}

\subsection{Correctness}
\label{F-subsec:disws-correctness}

The correctness rests on two facts. The first one confines the
visibility polygon of a point to one side of each of its two
windows. The second one turns this confinement into a
characterization of the conflict relation.

\begin{lemma}
\label{F-lem:shadow}
Every point $p \in S$ satisfies
$\vis(p) \subseteq \wv^-_r(p) \cap \wv^+_l(p)$.
\end{lemma}

\begin{proof}
We prove the inclusion $\vis(p) \subseteq \wv^-_r(p)$, and the other
inclusion is symmetric. The claim is trivial when $r(p) = v_n$,
so assume that $\RC(p)$ is a chord. Write $C = \RC(p)$.

We first place $p$ itself. Choose a point $b$ of $\eb$ with
$\ell(p) < b < r(p)$, which exists because $\ell(p) < r(p)$. The
point $b$ lies in $\wv^-_r(p)$ and not on $C$. The segment
$\seg{pb}$ lies in $\wv$, since $b$ is visible from $p$. It does
not meet $C$. Indeed, a common point of $\seg{pb}$ and $C$ would be
collinear with $p$ and with $\RA(p)$, and it would force $b$
onto the line through $p$ and $\RA(p)$, whence $b = r(p)$.
Hence $p$ and $b$ lie in the same subpolygon, and $p \in \wv^-_r(p)$.

Now let $q$ be a point of $\vis(p)$ that lies in the interior of
$\wv^+_r(p)$. The segment $\seg{pq}$ lies in $\wv$ and joins
$\wv^-_r(p)$ to the interior of $\wv^+_r(p)$, so it meets $C$ in
some point $y$. The three points $p$, $y$ and $q$ are collinear,
and $y$ lies on the line through $p$ and $\RA(p)$. Hence $q$
lies on that line too. The part of that line beyond $r(p)$ leaves
$\wv$, because the line meets $\eb$ transversally at $r(p)$ under
general position. Therefore $q$ lies on $\Rb(p)$, and no point
of $\Rb(p)$ lies in the interior of $\wv^+_r(p)$. This
contradiction shows that $\vis(p)$ misses the interior of
$\wv^+_r(p)$, which is the claim.
\end{proof}

\begin{lemma}
\label{F-lem:separation}
Two points $p$ and $q$ of $S$ satisfy
$\vis(p) \cap \vis(q) = \emptyset$ if and only if one of the following two conditions holds:
\begin{enumerate}
\item[(i)] $r(p) < \ell(q)$ and $\RA(p) \prec \LA(q)$;
\item[(ii)] $r(q) < \ell(p)$ and $\RA(q) \prec \LA(p)$.
\end{enumerate}
\end{lemma}

\begin{proof}
We record first that the inequality $r(p) < \ell(q)$ forces both
windows involved to be chords. It gives $r(p) \neq v_n$, since
$\ell(q)$ lies on $\eb$, and it gives $\ell(q) \neq v_0$ for the same
reason.

Assume condition~(i), and write $C = \RC(p)$ and $D = \LC(q)$.
The four endpoints of $C$ and $D$ are distinct. Read $\bd(\wv)$
clockwise from $v_0$ as in the proof of \cref{F-lem:anchor-order}.
The two hypotheses of~(i) place the four endpoints in the clockwise
order $\RA(p)$, $\LA(q)$, $\ell(q)$, $r(p)$. The chord
$C$ joins the first endpoint to the fourth, and the chord $D$ joins
the second to the third, so the two endpoint pairs do not
interleave and the chords are disjoint. Both endpoints of $D$ lie
on the boundary arc of $C$ that contains $v_n$, so
$D \subseteq \wv^+_r(p)$. The boundary of $\wv^+_l(q)$ consists of
$D$ and of the boundary arc from $\LA(q)$ to $\ell(q)$ through
$v_n$, and both lie in $\wv^+_r(p)$ and miss $C$. Hence
$\wv^+_l(q) \subseteq \wv^+_r(p)$ and
$\wv^+_l(q) \cap C = \emptyset$. The two subpolygons $\wv^-_r(p)$
and $\wv^+_r(p)$ meet exactly in $C$, so $\wv^-_r(p)$ and
$\wv^+_l(q)$ are disjoint. By \cref{F-lem:shadow} we get
$\vis(p) \subseteq \wv^-_r(p)$ and
$\vis(q) \subseteq \wv^+_l(q)$, and the two visibility polygons are
disjoint. Condition~(ii) is symmetric.

Conversely, assume that neither condition holds. We give a point of
$\vis(p) \cap \vis(q)$ in each remaining case.

Suppose first that $\I(p) \cap \I(q) \neq \emptyset$. A common point
of the two base intervals is visible from $p$ and from $q$, and the two visibility polygons meet.

Suppose next that $\I(p) \cap \I(q) = \emptyset$, and assume without
loss of generality that $r(p) < \ell(q)$. Both windows $\RC(p)$
and $\LC(q)$ are then chords. Since condition~(i) fails, we have
$\LA(q) \preceq \RA(p)$. If $\LA(q) = \RA(p)$,
then this common vertex lies on $\Rb(p)$ and on $\Lb(q)$, so
it is visible from $p$ and from $q$. If
$\LA(q) \prec \RA(p)$, then the four endpoints appear in
the clockwise order $\LA(q)$, $\RA(p)$, $\ell(q)$,
$r(p)$. The four endpoints interleave, so the chords cross by
\cref{F-obs:chord-cross}. A crossing point lies on $\Rb(p)$ and on
$\Lb(q)$, so it is visible from $p$ and from $q$. In both cases, the two
visibility polygons meet, which completes the proof.
\end{proof}


\begin{remark}
\label{F-rem:intervals-insufficient}
The base intervals alone do not decide the conflict relation. Two
points can have disjoint base intervals and still see a common
point, and by the proof above, this happens exactly when their two
windows cross. The upper side of the trapezoid records precisely
this case. Only the converse survives without it: disjoint
visibility polygons do force disjoint base intervals, because
$\I(p) = \vis(p) \cap \eb$ for every $p$.
\end{remark}

We now transfer the separation lemma to the trapezoids.

\begin{observation}
\label{F-obs:trapezoid-disjoint}
Two points $p$ and $q$ of $S$ satisfy $T(p) \cap T(q) = \emptyset$
if and only if condition~(i) or condition~(ii) of
\cref{F-lem:separation} holds.
\end{observation}

\begin{proof}
For a level $y$ in the closed strip bounded by $L_b$ and $L_t$, the
intersection of $T(p)$ with the horizontal line of level $y$ is a
closed interval $[a_p(y), b_p(y)]$, and both endpoints are affine
functions of $y$. Put $f = a_q - b_p$ and $g = a_p - b_q$, which
are again affine. The two intervals of level $y$ are disjoint if
and only if $f(y) > 0$ or $g(y) > 0$. The two conditions never hold
together, so the sets $\{f > 0\}$ and $\{g > 0\}$ are disjoint, and
each is relatively open in the strip. The strip is connected, so
these two sets cover it only if one of them is the whole strip.
Hence $T(p)$ and $T(q)$ are disjoint if and only if $f$ is positive
on the whole strip or $g$ is positive on the whole strip. An affine
function is positive on the strip if and only if it is positive on
both bounding lines. On $L_b$ the inequality $f > 0$ reads
$r(p) < \ell(q)$, and on $L_t$ it reads
$\RA(p) \prec \LA(q)$, because $\beta$ and $\tau$ preserve
the two orders. These are the two parts of condition~(i), and $g$
gives condition~(ii) in the same way.
\end{proof}

\begin{theorem}
\label{F-thm:trapezoid-graph}
Let $(\wv, S)$ be an instance of \textsc{Discrete Witness Set}. The
visibility intersection graph $G_S$ is the intersection graph of the trapezoids
$\{T(s) : s \in S\}$. In particular, $G_S$ is a trapezoid graph, and
the family $\{T(s) : s \in S\}$ is a trapezoid representation of it.
\end{theorem}

\begin{proof}
Let $s$ and $s'$ be two distinct points of $S$. By
\cref{F-lem:separation} the visibility polygons of $s$ and $s'$ are
disjoint if and only if condition~(i) or condition~(ii) holds, and
by \cref{F-obs:trapezoid-disjoint} the same two conditions
characterize disjointness of $T(s)$ and $T(s')$. Hence $s$ and $s'$
are adjacent in $G_S$ if and only if $T(s)$ and $T(s')$ intersect.
The trapezoids are well defined by \cref{F-lem:anchor-order}, and they
have their parallel sides on the two fixed lines $L_b$ and $L_t$, so
they form a trapezoid representation in the sense
of~\cite{DaganGP88,FelsnerMW97}.
\end{proof}

We show in \cref{F-subsec:disws-expressive} that this class strictly
contains the interval graphs, so the trapezoid structure is not an
artifact of the analysis.

\begin{corollary}
\label{F-cor:witness-mis}
A subset $Q \subseteq S$ is a witness set if and only if the
trapezoids $\{T(s) : s \in Q\}$ are pairwise disjoint. The maximum
witness subsets of $S$ are exactly the maximum independent sets of
the trapezoid graph of \cref{F-thm:trapezoid-graph}.
\end{corollary}

\begin{proof}
Combine \cref{F-def:disws} with \cref{F-thm:trapezoid-graph}.
\end{proof}

\begin{theorem}
\label{F-thm:disws-correct}
\Cref{F-alg:disws-wvp} returns a maximum witness subset of $S$.
\end{theorem}

\begin{proof}
The algorithm computes the anchors of \cref{F-def:prelim-anchors}, and it
computes $\ell(s)$ and $r(s)$ from them by
\cref{F-obs:anchor-collinear}. It then builds the trapezoids of
\cref{F-def:trapezoid}, which represent $G_S$ by
\cref{F-thm:trapezoid-graph}. A maximum independent set of that
representation is a maximum witness subset of $S$ by
\cref{F-cor:witness-mis}, and the algorithm of~\cite{FelsnerMW97} returns one.
\end{proof}

\subsection{Running Time}
\label{F-subsec:disws-runtime}

\begin{theorem}
\label{F-thm:disws-runtime}
Let $\wv$ be a weak visibility polygon with $n$ vertices and let
$S$ be a set of $m$ points of $\wv$. \Cref{F-alg:disws-wvp} computes
a maximum witness subset of $S$ in $\OO(n + m \log (n+m))$ time and
$\OO(n + m)$ space.
\end{theorem}

\begin{proof}
We account for the four phases of \cref{F-alg:disws-wvp}.

The triangulation of $\wv$ takes $\OO(n)$ time~\cite{Chazelle91}.
Each of the two shortest path maps is built in $\OO(n)$ time from
the triangulation~\cite{GHLST87}, and each has $\OO(n)$ cells, so a
point location structure with $\OO(\log n)$ query time is built in
$\OO(n)$ additional time. This phase costs $\OO(n)$ time.

The anchor phase performs two-point location queries per point of
$S$, at $\OO(\log n)$ time each. Each cell of a shortest path map
stores the last vertex on the shortest path from the source to the
points of that cell, which is the anchor by \cref{F-def:prelim-anchors}.
Given the anchor, the endpoint $\ell(s)$ or $r(s)$ is one
intersection of a ray with the supporting line of $\eb$, computed in
$\OO(1)$ time by \cref{F-obs:anchor-collinear}. This phase costs
$\OO(m \log n)$ time.

The construction phase writes down one trapezoid per point of $S$ in
$\OO(1)$ time each, and it sorts the $2m$ endpoints on $L_b$ and the
$2m$ endpoints on $L_t$. This phase costs $\OO(m \log m)$ time.

The last phase computes a maximum independent set of a trapezoid
graph on $m$ trapezoids from its representation, in $\OO(m \log m)$
time~\cite{FelsnerMW97}.

The four phases together take $\OO(n + m \log n + m \log m)$ time,
which is $\OO(n + m \log (n+m))$. The triangulation, the two maps
and their point location structures use $\OO(n)$ space, the
trapezoids use $\OO(m)$ space, and the algorithm
of~\cite{FelsnerMW97} uses $\OO(m)$ space.
\end{proof}

\begin{corollary}
\label{F-cor:disws-runtime}
If the number of input points satisfies $m = \OO(n)$, then a maximum
witness subset of $S$ is computed in $\OO(n \log n)$ time.
\end{corollary}

\begin{remark}
\label{F-rem:no-graph}
\Cref{F-alg:disws-wvp} never builds the visibility intersection graph $G_S$
explicitly. The graph can have $\Theta(m^2)$ edges, for instance
when every point of $S$ sees a common point of $\eb$, so the bound
of \cref{F-thm:disws-runtime} is below the size of the object it
optimizes over.
\end{remark}

\subsection{Optimality of the Running Time}
\label{F-subsec:disws-lower}

The near-linear bound of \cref{F-cor:disws-runtime} is best possible in the algebraic decision-tree model, already when the number of input points is proportional to the size of the polygon. The hardness comes
from the witness side of the problem: a maximum witness subset is a
maximum independent set of the visibility intersection graph, and the visibility intersection graphs
of \cref{F-thm:trapezoid-graph} are rich enough to contain every
permutation graph.

\begin{figure}[htbp]
\centering
\includegraphics[scale=0.9]{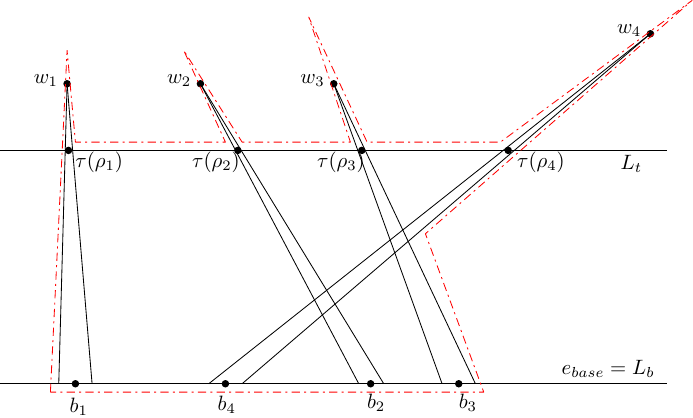}
\caption{The construction of \cref{F-lem:perm-realization} for
$\pi = (1,4,2,3)$. Each witness $w_i$ reaches the base only through
its gateway reflex vertex, its common anchor
$\rho_i = \LA(w_i) = \RA(w_i)$, so its trapezoid collapses to the
segment from the base point $b_i$ up to $\tau(\rho_i)$ on $L_t$ (drawn
as a thin pocket). The base points appear in the order of $\pi$ and
the top points in the identity order, so two segments cross exactly at
an inversion of $\pi$; here $w_4$ conflicts with $w_2$ and $w_3$. A
maximum independent set is a longest increasing subsequence of $\pi$,
for instance $\{w_1, w_2, w_3\}$.}
\label{F-fig:perm-wvp}
\end{figure}

\begin{lemma}[Permutation Graphs are Visibility Intersection Graphs]
\label{F-lem:perm-realization}
For every permutation $\pi$ of $\{1, \dots, k\}$ there is an instance
$(\wv_\pi, S_\pi)$ of \textsc{Discrete Witness Set} in which $\wv_\pi$
has $\OO(k)$ vertices, $\lvert S_\pi \rvert = k$, and the visibility
intersection graph $G_{S_\pi}$ is isomorphic to the permutation graph
of $\pi$. The instance is produced from $\pi$ in $\OO(k)$ time.
\end{lemma}

\begin{proof}
Take the base $\eb$ to be the horizontal segment from $u = (0,0)$ to
$v = (k+1,0)$. We build the upper chain of $\wv_\pi$ so that it
carries, in clockwise order, $k$ reflex vertices
$\rho_1, \rho_2, \dots, \rho_k$ of strictly increasing abscissa, with
one shallow pocket hanging behind each $\rho_i$. Inside the pocket
behind $\rho_i$ we place a single input point $w_i$, so deep in the
pocket that $\rho_i$ is the only vertex of $\wv_\pi$ through which $w_i$
sees the base; that is, the pocket has $\rho_i$ as its sole gateway to
$\eb$. Both shortest paths $\spl(w_i)$ and $\spr(w_i)$ then leave the
pocket through $\rho_i$, so
\[
  \LA(w_i) \;=\; \RA(w_i) \;=\; \rho_i .
\]
By \cref{F-obs:anchor-collinear} the endpoints $\ell(w_i)$ and $r(w_i)$
both lie on the line $\seg{w_i\,\rho_i}$; hence they coincide with the
single point $b_i$ where that line, extended, meets $\eb$, and
$\I(w_i) = \{b_i\}$.

The pocket may be aimed freely. Choosing where $w_i$ sits behind
$\rho_i$ turns the ray $\seg{w_i\,\rho_i}$ to any direction, so $b_i$
can be placed at any prescribed abscissa on $\eb$, independently of the
other pockets: the pockets are separated horizontally, and the region below the reflex vertices is empty, so the aiming rays may cross there
without any pocket obstructing another's single view of the base. We
aim them so that the base points realize the order of $\pi$,
\[
  x(b_i) < x(b_j) \quad\Longleftrightarrow\quad \pi(i) < \pi(j).
\]
Each pocket contributes $\OO(1)$ vertices, so $\wv_\pi$ has $\OO(k)$
vertices and is written down from $\pi$ in $\OO(k)$ time; it is weakly
visible from $\eb$ because each pocket is seen from its base point
$b_i$ and the remainder of $\wv_\pi$ lies directly above $\eb$.

It remains to read off the visibility intersection graph. Recall the
two horizontal lines $L_b$, $L_t$ and the order-preserving maps
$\beta \colon \eb \to L_b$ and $\tau$ of
\cref{F-subsec:disws-construction}. Because
$\LA(w_i) = \RA(w_i) = \rho_i$, the upper side of the trapezoid
$T(w_i)$ is the single point $\tau(\rho_i)$, and because
$\ell(w_i) = r(w_i) = b_i$ its lower side is the single point
$\beta(b_i)$. Every trapezoid therefore
degenerates to a segment
\[
  T(w_i) \;=\; \seg{\beta(b_i)\;\tau(\rho_i)}
\]
joining a point of $L_b$ to a point of $L_t$. The reflex vertices were
placed with increasing abscissa, so their clockwise indices increase
with $i$ and the top endpoints $\tau(\rho_1), \dots, \tau(\rho_k)$
appear in this order along $L_t$; the bottom endpoints
$\beta(b_1), \dots, \beta(b_k)$ appear along $L_b$ in the order of
$\pi$. Two segments spanning a pair of parallel lines cross if and only if their endpoints occur in opposite orders on the two lines, so
$T(w_i)$ and $T(w_j)$ cross exactly when $i < j$ and $\pi(i) > \pi(j)$,
that is, exactly when $\{i,j\}$ is an inversion of $\pi$. By
\cref{F-thm:trapezoid-graph} the visibility intersection graph
$G_{S_\pi}$ is the intersection graph of the family $\{T(w_i)\}$, which
is thus the permutation graph of $\pi$.

Finally, replacing the single aiming ray of each pocket by the two
bounding rays of a sufficiently narrow pocket makes each $\I(w_i)$ a
short interval and each $T(w_i)$ a thin trapezoid, without altering
which pairs cross; this satisfies the general position assumption of
\cref{F-sec:prelims} while preserving $G_{S_\pi}$. \Cref{F-fig:perm-wvp}
shows the construction for $\pi = (1,4,2,3)$.
\end{proof}

\begin{theorem}[Lower Bound]
\label{F-thm:disws-lower}
In the algebraic decision-tree model, \textsc{Discrete Witness Set}
requires $\Omega(k \log k)$ operations in the worst case on instances
with $\lvert S \rvert = k$ points in a weak visibility polygon of
$\OO(k)$ vertices.
\end{theorem}

\begin{proof}
A maximum independent set of the permutation graph of $\pi$ has size
equal to the length of a longest increasing subsequence of $\pi$, and
computing this length requires $\Omega(k \log k)$ comparisons in the
worst case~\cite{Fredman75}; the bound holds in the algebraic
decision-tree model, as does the equivalent bound obtained by a
reduction from \textsc{Element
Uniqueness}~\cite{DBLP:conf/stoc/Ben-Or83,DBLP:books/sp/PreparataS85}.
Suppose an algorithm solved \textsc{Discrete Witness Set} in
$o(k \log k)$ operations on the stated instances. Given $\pi$, build
$(\wv_\pi, S_\pi)$ in $\OO(k)$ time by \cref{F-lem:perm-realization} and run the algorithm on it. By \cref{F-cor:witness-mis} the size of its
output equals a maximum independent set of $G_{S_\pi}$, hence the
longest increasing subsequence length of $\pi$, computed in
$o(k \log k)$ operations, a contradiction.
\end{proof}

\begin{corollary}
\label{F-cor:disws-tight}
For instances with $m = \Theta(n)$, \textsc{Discrete Witness Set} in a
weak visibility polygon is solvable in $\Theta(n \log n)$ time, and no algebraic decision-tree algorithm does better.
\end{corollary}

\begin{proof}
The upper bound is \cref{F-cor:disws-runtime}. For the lower bound,
apply \cref{F-thm:disws-lower} with $k = \Theta(n)$: the polygon then has
$\OO(n)$ vertices and $m = \Theta(n)$ points, and
$k \log k = \Theta(n \log n)$.
\end{proof}

\subsection{The Class of Visibility Intersection Graphs}
\label{F-subsec:disws-expressive}

By \cref{F-thm:trapezoid-graph} every visibility intersection graph is a
trapezoid graph. We now locate this class more precisely. It
contains two familiar families, the interval graphs and the
permutation graphs, and it contains each of them properly. The
permutation graphs were already realized in \cref{F-lem:perm-realization}; we treat the interval graphs next.

\begin{lemma}[Interval Graphs are Visibility Intersection Graphs]
\label{F-lem:interval-realization}
For every interval graph $H$ there is an instance $(\wv, S)$ of
\textsc{Discrete Witness Set} whose visibility intersection graph
$G_S$ is isomorphic to $H$. The polygon $\wv$ has $\OO(n)$ vertices,
where $n$ is the number of vertices of $H$.
\end{lemma}

\begin{proof}
Fix an interval representation of $H$: closed intervals
$I_1, \dots, I_n$ on a horizontal line, with distinct endpoints, such
that two vertices of $H$ are adjacent if and only if the corresponding
intervals meet. Write $I_i = [a_i, b_i]$. Say that $I_i$ is
\emph{left of} $I_j$ when $b_i < a_j$, that is, when $I_i$ lies
entirely to the left of $I_j$. This relation is a strict partial order, and we fix a linear extension $\sigma$ of it, so that
$b_i < a_j$ implies $\sigma(i) < \sigma(j)$.

We reuse the pocket construction of \cref{F-lem:perm-realization}, with
the single change that the pockets are wider. Put $\eb$ on a
horizontal segment and build the upper chain of $\wv$ so that it
carries $n$ pairwise disjoint shallow pockets, arranged from left to
right in the order given by $\sigma$. The mouth of the pocket that
holds vertex $i$ has a left reflex vertex $\lambda_i$ and a right
reflex vertex $\rho_i$. Place a single input point $w_i$ so deep in
that pocket that the pocket is its only gateway to $\eb$. Then
$\spl(w_i)$ leaves the pocket around $\lambda_i$ and $\spr(w_i)$ leaves
it around $\rho_i$, so
\[
  \LA(w_i) = \lambda_i, \qquad \RA(w_i) = \rho_i .
\]
Hence $w_i$ sees on $\eb$ exactly the interval bounded by the rays
$\seg{w_i\,\lambda_i}$ and $\seg{w_i\,\rho_i}$, namely
$\I(w_i) = [\ell(w_i), r(w_i)]$, where $\ell(w_i)$ and $r(w_i)$ are the
points in which these two rays meet $\eb$. As in
\cref{F-lem:perm-realization} the pockets are separated horizontally, and the region below the reflex vertices is empty, so the two rays of each
pocket may be aimed independently of the other pockets. Aim them so
that $\ell(w_i) = a_i$ and $r(w_i) = b_i$, giving $\I(w_i) = I_i$.
Each pocket contributes $\OO(1)$ vertices, so $\wv$ has $\OO(n)$
vertices; it is weakly visible from $\eb$ because each pocket is seen
from its own base interval, and the rest of $\wv$ lies directly above
$\eb$.

It remains to identify $G_S$. Recall the two parallel lines $L_b$ and
$L_t$ and the order-preserving maps $\beta$ and $\tau$ of
\cref{F-subsec:disws-construction}. The pockets are disjoint and
ordered by $\sigma$, so on $L_t$ the upper sides
$\seg{\tau(\lambda_i)\,\tau(\rho_i)}$ of the trapezoids $T(w_i)$ are
pairwise disjoint and appear in the order of $\sigma$; on $L_b$ the
lower sides are the intervals $\I(w_i) = I_i$. Take two vertices $i$
and $j$ with $\sigma(i) < \sigma(j)$, so the upper side of $T(w_i)$
lies entirely left of the upper side of $T(w_j)$. Two trapezoids
spanning $L_b$ and $L_t$ are disjoint if and only if one lies entirely
left of the other; since the upper sides are already ordered, $T(w_i)$
and $T(w_j)$ are disjoint if and only if their lower sides are ordered
the same way, that is $b_i < a_j$. Now compare with $H$. If $I_i$ and
$I_j$ are disjoint, they are comparable in the left-of order; as
$\sigma$ extends that order and $\sigma(i) < \sigma(j)$, the interval
$I_i$ is left of $I_j$, so $b_i < a_j$ and the trapezoids are disjoint.
If $I_i$ and $I_j$ meet, their lower sides overlap, and the trapezoids
are not disjoint. By \cref{F-thm:trapezoid-graph}, the graph $G_S$ is the
intersection graph of the family $\{T(w_i)\}$, so $w_i$ and $w_j$ are
adjacent in $G_S$ if and only if $I_i$ and $I_j$ meet, which is
adjacency in $H$. Therefore, $G_S$ is isomorphic to $H$.
\end{proof}

Both containments are strict, and a single small graph separates the interval graphs from the visibility intersection graphs.

Together with \cref{F-lem:perm-realization} and \cref{F-lem:interval-realization}, the trapezoid representation locates the visibility intersection graphs between two familiar classes.

\begin{theorem}[Location of the Visibility Intersection Graphs]
\label{F-thm:VIGgraphclasslocation}
Let $\mathsf{VIG}$ denote the class of visibility intersection graphs of
instances of \textsc{Discrete Witness Set} in weak visibility polygons,
and let $\mathsf{Int}$, $\mathsf{Perm}$ and $\mathsf{Trap}$ denote the
classes of interval graphs, permutation graphs, and trapezoid graphs
respectively. Then
\begin{equation}
\label{F-eq:vig-location}
  \mathsf{Int} \subsetneq \mathsf{VIG},
  \qquad
  \mathsf{Perm} \subsetneq \mathsf{VIG},
  \qquad
  \mathsf{VIG} \subseteq \mathsf{Trap}.
\end{equation}
\end{theorem}

\begin{proof}
The containment $\mathsf{Int} \subseteq \mathsf{VIG}$ is
\cref{F-lem:interval-realization}, and the containment
$\mathsf{Perm} \subseteq \mathsf{VIG}$ is \cref{F-lem:perm-realization}.
The containment $\mathsf{VIG} \subseteq \mathsf{Trap}$ is
\cref{F-thm:trapezoid-graph}.

It remains to prove that the first two containments are proper. The
interval graphs and the permutation graphs are incomparable, so neither
class contains the other~\cite{GOLUMBIC20041}. Choose a permutation
graph $G_1$ that is not an interval graph, and choose an interval graph $G_2$ that is not a permutation graph. By \cref{F-lem:perm-realization}
the graph $G_1$ belongs to $\mathsf{VIG}$, and it does not belong to
$\mathsf{Int}$. By \cref{F-lem:interval-realization} the graph $G_2$
belongs to $\mathsf{VIG}$, and it does not belong to $\mathsf{Perm}$.
\end{proof}

\begin{remark}
\label{F-rem:vig-trap-open}
Whether the third containment in \cref{F-eq:vig-location} is proper is
open. A positive answer would characterize the conflict structure of
\textsc{Discrete Witness Set} exactly.
\end{remark}

\providecommand{\lo}{\lhd}                    
\providecommand{\tk}{\mathsf{trapShadow}}     
\providecommand{\RUC}{\mathsf{RUC}}           
\providecommand{\tp}{\mathcal{X}}             
\providecommand{\cand}{\mathcal{C}}           
\providecommand{\anc}{\mathcal{A}}            
\providecommand{\rst}{r^{*}}                  
\providecommand{\level}{h}                    

\section{\textsc{Continuous Witness Set Problem} in $\wvp$}
\label{F-sec:wvp-max-witness}
\label{F-sec:ws-correct}

We give an exact polynomial-time algorithm that computes a maximum witness set for the {\sc Continuous Witness Set Problem} $(\conWSP)$, prove its correctness, and bound its running time in \cref{F-sec:ws-time}.

Throughout, $\wv$ is a simple polygon with $n$ vertices that is weakly
visible from its base edge $\eb$, and $\rho$ is the number of reflex
vertices of $\wv$. The base has endpoints $u = v_0$ and $v = v_n$, and
we assume the general position of \cref{F-sec:prelims}. We reuse the
visibility polygon $\vis(p)$, the base interval
$\I(p) = [\ell(p), r(p)] = \vis(p) \cap \eb$, the anchors $\LA(p)$ and
$\RA(p)$, and the clockwise order $\prec$ on the upper chain. The
endpoints $\ell(p)$ and $r(p)$ follow from the two anchors in
$\OO(\log n)$ time after linear-time preprocessing~\cite{GHLST87}.
\Cref{F-lem:separation} and \cref{F-obs:trapezoid-disjoint} were stated for
input points, but their proofs use only general position, so we apply
them to arbitrary points of $\wv$.

A \emph{witness set} is a set of points of $\wv$ that are pairwise
witnesses (\cref{F-def:witness-set}), meaning that no point of $\wv$ sees
two of them. A \emph{maximum} witness set is one of the largest cardinality. The witness relation no longer lives on $\eb$, yet the
base intervals still order the candidates soundly.

\begin{observation}
\label{F-obs:vis-implies-interval}
If $\vis(w) \cap \vis(w') = \emptyset$, then
$\I(w) \cap \I(w') = \emptyset$.
\end{observation}

\begin{proof}
Since $\I(p) = \vis(p) \cap \eb$ for every $p \in \wv$, any common
point of $\I(w)$ and $\I(w')$ lies in both visibility polygons.
\end{proof}

The converse fails. Interval disjointness is only a necessary filter,
and full visibility disjointness must be checked separately. This is
what makes the continuous problem harder than the discrete one.

The proof of correctness has five parts. A witness set is a chain in a
partial order. The same points can always extend a shorter chain that ends further left than a longer one. Every witness can be
moved to the boundary without changing the size of the witness set. On
the boundary, every witness can be moved to one of finitely many
candidate points. A small table over anchors and chain lengths then
finds the longest chain among these candidates.

\subsection{Witness Sets Are Chains}
\label{F-subsec:ws-chains}

The separation lemma says that two visibility polygons are disjoint
exactly when one trapezoid lies to the left of the other on both
lines. We give this relation a name.

\begin{definition}[left of]
\label{F-def:ws-leftof}
Let $p$ and $q$ be points of $\wv$. We write $p \lo q$ and say that
$p$ is \emph{left of} $q$ when
\[
  r(p) < \ell(q)
  \qquad\text{and}\qquad
  \RA(p) \prec \LA(q).
\]
\end{definition}

The next lemma turns the witness problem into a question about
chains.

\begin{lemma}
\label{F-lem:ws-chain}
The relation $\lo$ is a strict partial order on the points of $\wv$.
Two points $p$ and $q$ satisfy $\vis(p) \cap \vis(q) = \emptyset$ if
and only if $p \lo q$ or $q \lo p$. A set $W \subseteq \wv$ is a
witness set if and only if it is a chain under $\lo$.
\end{lemma}

\begin{proof}
No point $p$ satisfies $p \lo p$, because $\ell(p) \le r(p)$. Let
$p \lo q$ and $q \lo s$. Then
\[
  r(p) < \ell(q) \le r(q) < \ell(s).
\]
Every point $q$ satisfies $\LA(q) \preceq \RA(q)$, so
\[
  \RA(p) \prec \LA(q) \preceq \RA(q) \prec \LA(s).
\]
Hence $p \lo s$, and the relation is transitive. The second claim
restates \cref{F-lem:separation}. For the third claim, a set is a
witness set exactly when every two of its members are comparable
under $\lo$. This is the definition of a chain.
\end{proof}

A witness set is therefore a chain, and a maximum witness set is a longest chain. Transitivity also justifies the test used by the
algorithm. It suffices to compare a new point with the last point of
a chain.

\begin{lemma}
\label{F-lem:ws-order}
Let $x_1 \lo x_2 \lo \dots \lo x_m$ be a chain and let $w$ be a point
with $x_m \lo w$. Then $x_j \lo w$ for every index $j$, and the set
$\{x_1, \dots, x_m, w\}$ is a witness set.
\end{lemma}

\begin{proof}
Transitivity gives $x_j \lo x_m \lo w$ for every $j < m$. The set is
a chain, so it is a witness set by \cref{F-lem:ws-chain}.
\end{proof}

Disjointness of two visibility polygons is a symmetric condition. A
new point may be disjoint from the last point of a chain and still lie
to its left. One extra comparison on the base, rules this out.

\begin{lemma}
\label{F-lem:ws-guard}
Let $x$ and $w$ satisfy $\vis(x) \cap \vis(w) = \emptyset$ and
$\ell(w) > r(x)$. Then $x \lo w$.
\end{lemma}

\begin{proof}
By \cref{F-lem:ws-chain} we have $x \lo w$ or $w \lo x$. Suppose that
$w \lo x$. Then $r(w) < \ell(x)$. But $r(w) \ge \ell(w) > r(x) \ge
\ell(x)$. This is a contradiction, so $x \lo w$.
\end{proof}

\subsection{Substitution}
\label{F-subsec:ws-domination}

A chain leaves room to its right through two quantities. One is the
right endpoint of the last base interval. The other is the right
anchor of the last point. A point that is no larger in both can take
the place of the last point without losing any successor.

\begin{lemma}
\label{F-lem:ws-dom}
Let $x \lo q$, and let $x'$ be a point with $r(x') \le r(x)$ and
$\RA(x') \preceq \RA(x)$. Then $x' \lo q$.
\end{lemma}

\begin{proof}
From $x \lo q$ we get $r(x) < \ell(q)$ and $\RA(x) \prec \LA(q)$.
Hence
\[
  r(x') \le r(x) < \ell(q)
  \qquad\text{and}\qquad
  \RA(x') \preceq \RA(x) \prec \LA(q),
\]
which is the statement $x' \lo q$.
\end{proof}

\Cref{F-lem:ws-dom} says that, among chains of the same length whose last
points share a right anchor; only the smallest right endpoint matters.
This is the reason the table in \cref{F-subsec:ws-table} is small.

\subsection{Witnesses Can Be Moved to the Boundary}
\label{F-subsec:ws-boundary}

The algorithm's candidates lie on the boundary. A maximum
witness set may contain interior points. We show that every interior
witness can be pushed to the boundary along a ray through its right anchor;
see \cref{F-fig:ws-push}. The push leaves $r$ and the right anchor unchanged
and shrinks the trapezoid.

\begin{definition}[Trapezoid Shadow]
\label{F-def:ws-kernel}
The \emph{trapezoid shadow} of a point $w$ is the set $\tk(w)$ of all
points $x \in \wv$ whose trapezoid (\cref{F-def:trapezoid}) satisfies
$T(x) \subseteq T(w)$.
\end{definition}

Replacing $w$ by a point of $\tk(w)$ preserves every relation $\lo$ in which
$w$ takes part; the next lemma records the consequence we use.

\begin{lemma}
\label{F-lem:ws-kernel}
Let $x \in \tk(w)$ and let $z$ be a point with
$\vis(w) \cap \vis(z) = \emptyset$. Then
$\vis(x) \cap \vis(z) = \emptyset$.
\end{lemma}

\begin{proof}
By \cref{F-obs:trapezoid-disjoint} the trapezoids $T(w)$ and $T(z)$ are
disjoint. Since $T(x) \subseteq T(w)$, the trapezoids $T(x)$ and
$T(z)$ are disjoint as well. The same observation gives
$\vis(x) \cap \vis(z) = \emptyset$.
\end{proof}

We now name the boundary point to which we push a witness. Let
$w \in \wv$ and let $\RA(w)$ be a reflex vertex. By
\cref{F-obs:anchor-collinear} the points $r(w)$, $\RA(w)$, $w$ lie on one
ray in this order, and we set
\[
  \RUC(w) \;=\; \win(r(w), \RA(w)),
\]
the first point of $\bd(\wv)$ on that ray after $\RA(w)$. The point $\RUC(w)$
lies on the upper chain, because the ray moves away from the base.

\begin{figure}[htbp]
    \centering
    \includegraphics[scale=0.9]{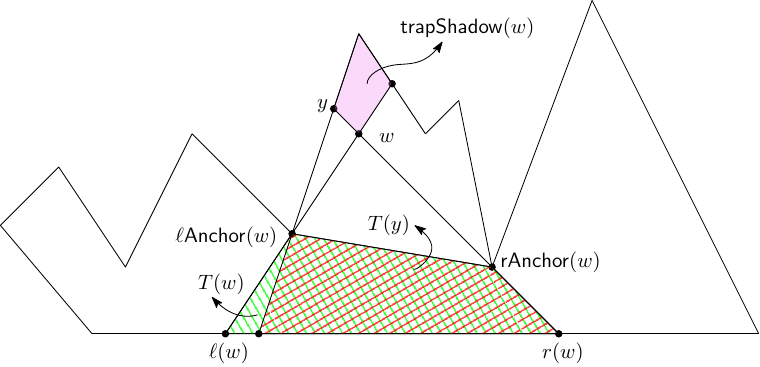}
    \caption{The push to the boundary. The point $w$ is pushed along the
      ray from $\RA(w)$ through $w$ to the boundary point $\RUC(w)$. The
      right endpoint $r(w)$ and the right anchor $\RA(w)$ do not change, see \cref{F-lem:ws-push}.}
    \label{F-fig:ws-push}
\end{figure}

\begin{lemma}
\label{F-lem:ws-push}
Let $w$ be a point of $\wv$ whose right anchor is a reflex vertex, and
let $y = \RUC(w)$. Then
\begin{gather*}
  r(y) = r(w) \qquad\text{and}\qquad \RA(y) = \RA(w), \\
  \ell(y) \ge \ell(w) \qquad\text{and}\qquad \LA(y) \succeq \LA(w).
\end{gather*}
Hence $T(y) \subseteq T(w)$, and $y$ lies in $\tk(w)$.
\end{lemma}

\begin{proof}
Write $L$ for the line through $w$ and $\RA(w)$. The segment from $y$
to $r(w)$ lies on $L$. It is the union of the segments from $y$ to
$w$, from $w$ to $\RA(w)$, and from $\RA(w)$ to $r(w)$. Each of these
lies in $\wv$. Hence $y$ sees $r(w)$, and the line of sight passes
through the reflex vertex $\RA(w)$. A line of sight through a reflex
vertex bounds the visible part of the base. So $r(w)$ is an endpoint
of $\I(y)$ and $\RA(w)$ is the anchor at that endpoint. The vertex
$\RA(w)$ blocks the same side of $L$ for $y$ as for $w$, because both
points lie on $L$ beyond $\RA(w)$. This proves the first two
equalities.

We now prove $\ell(y) \ge \ell(w)$. Write $L'$ for the line through
$w$ and $\LA(w)$, and write $c_0 = \LC(w)$ for the left window. By
\cref{F-obs:window} the window $c_0$ is a chord. It splits $\wv$
into the part $\wv^-_l(w)$ that contains $v_0$ and the part
$\wv^+_l(w)$ that contains $v_n$. The point $w$ lies in
$\wv^+_l(w)$, and the segment from $w$ to $y$ meets $L'$ only at $w$.
Hence $y$ lies in $\wv^+_l(w)$ as well.

Suppose that $y$ sees a base point $b$ with $b < \ell(w)$. The point
$b$ lies in $\wv^-_l(w)$. So the segment from $y$ to $b$ crosses the
chord $c_0$ at a point $c$. We first show that $c \ne \LA(w)$. The
two edges of $\wv$ at $\LA(w)$ lie in one closed half plane of $L'$,
because the interior angle at $\LA(w)$ contains both directions of
$L'$. The base points left of $\ell(w)$ lie in the same half plane,
and so does $y$. If $c = \LA(w)$, then the ray from $y$ through
$\LA(w)$ leaves this half plane at $\LA(w)$ and meets the base on the
other side, that is, at a point right of $\ell(w)$. This contradicts
$b < \ell(w)$. So $c$ lies in the relative interior of $c_0$, and
$\LA(w)$ lies between $w$ and $c$ on $L'$.

Consider the triangle $\Delta$ with corners $y$, $w$, and $c$. Its
three sides lie in $\wv$. The boundary of $\wv$ meets these sides
only at $y$ and at $\LA(w)$. The two edges at $\LA(w)$ enter the
interior of $\Delta$. Follow the boundary of $\wv$ from $\LA(w)$
along either edge. It cannot cross a side of $\Delta$ so that it can leave $\Delta$ only through the corner $y$. Hence, the whole boundary
of $\wv$ lies in $\Delta$. But the base point $\ell(w)$ lies on $L'$
beyond $c$, outside $\Delta$. This is a contradiction. So $y$ sees
no base point to the left of $\ell(w)$, and $\ell(y) \ge \ell(w)$.

Finally, we compare the left anchors. The shortest path from $y$ to
$v_0$ starts in $\wv^+_l(w)$ and ends in $\wv^-_l(w)$, so it crosses
$c_0$. If its first bend comes before the crossing, then $\LA(y)$ is
a reflex vertex on the boundary of $\wv^+_l(w)$ other than $\LA(w)$,
and this boundary arc follows $\LA(w)$ in clockwise order. If the
first segment of the path reaches $c_0$, then it reaches $\LA(w)$,
because the argument above rules out a first segment that ends in the
relative interior of $c_0$. In both cases $\LA(y) \succeq \LA(w)$.

The trapezoid $T(y)$ has upper side
$\seg{\tau(\LA(y))\,\tau(\RA(w))}$ and lower side
$\seg{\beta(\ell(y))\,\beta(r(w))}$. Both sides are contained in the
corresponding sides of $T(w)$. Convex hulls of nested sides on two
parallel lines are nested. So $T(y) \subseteq T(w)$.
\end{proof}

When $\RA(w)$ is the base endpoint $v_n$, the point $w$ sees $v_n$.
The same proof works with the ray from $v_n$ through $w$, and we use
the same name $\RUC(w)$ for the resulting boundary point.

\begin{corollary}
\label{F-cor:ws-boundary}
Every witness set has the same size as a witness set that lies in
$\bd(\wv)$ and whose members are the points $\RUC(w)$ for the members
$w$ of the original set.
\end{corollary}

\begin{proof}
Let $W$ be a witness set and let $W'$ be the set of points $\RUC(w)$
with $w \in W$. Let $w_1$ and $w_2$ be two members of $W$. Their
visibility polygons are disjoint. By \cref{F-lem:ws-push} and
\cref{F-lem:ws-kernel}, the point $\RUC(w_1)$ has a visibility polygon
disjoint from $\vis(w_2)$. A second application gives that
$\RUC(w_1)$ and $\RUC(w_2)$ have disjoint visibility polygons. In
particular, they are distinct points. So $W'$ is a witness set of the
same size.
\end{proof}

\subsection{Discretization}
\label{F-subsec:ws-discrete}

The push lands on a window endpoint in the sense of
\cref{F-def:prelim-window}: the base point $r(w)$ sees $\RA(w)$, because
$\RA(w)$ lies on the segment from $w$ to $r(w)$ and this segment lies
in $\wv$, and the ray grazes the reflex vertex $\RA(w)$, so it supports
a window of $\vis(r(w))$ whose far endpoint is $\RUC(w)$. We now describe the finite set of static candidates that the algorithm uses in
addition.

We next fix the finite set of points at which the anchors of a boundary point can change.

\begin{definition}[Transition Points]
\label{F-def:ws-transition}
Let $Q$ be the set consisting of the reflex vertices of $\wv$ and the
two endpoints of $\eb$. The set $\tp$ of \emph{transition points}
consists of the points $\win(a, b)$ over all ordered pairs $(a, b)$ of
distinct points of $Q$ for which \cref{F-def:prelim-window} defines this
point.
\end{definition}

The set $\tp$ has at most $(\rho + 2)(\rho + 1)$ members. A point of the upper chain that is not a vertex and not a transition point has
the same anchors as its neighbors on the chain.

\begin{lemma}
\label{F-lem:ws-pieces}
Remove the vertices and the transition points from the upper chain.
Each remaining open arc is called a \emph{piece}. On a piece, the
anchors $\LA$ and $\RA$ are constant. Along a piece, the functions
$\ell$ and $r$ are continuous, and each is either constant or strictly
monotone. When both are strictly monotone, they increase in the same
direction.
\end{lemma}

\begin{proof}
Let $x$ move along a piece. The left anchor $\LA(x)$ is the first
reflex vertex on the shortest path from $x$ to $v_0$. It changes only
when $x$ becomes collinear with $\LA(x)$ and the next vertex on that
path, or with $\LA(x)$ and $v_0$. At such a point $x$ is a transition
point, because $x$ sees $\LA(x)$ and the ray through the two vertices
enters $\wv$ there. The same holds for the right anchor. So both
anchors are constant on the piece. The piece lies on one edge of
$\wv$. The value $\ell(x)$ is the image of $x$ under the central
projection from the edge to the baseline with center $\LA(x)$, and
$r(x)$ is the image under the projection with center $\RA(x)$. A
central projection between two lines whose center is not on either line
is continuous and strictly monotone. Both centers lie between the edge
and the base, so both projections reverse the orientation. Hence, the
two functions increase in the same direction. When an anchor is a base
endpoint, the corresponding function is constant.
\end{proof}

\begin{figure}[htbp]
\centering
\captionsetup[subfigure]{skip=1pt}

\begin{subfigure}{\textwidth}
\centering
\includegraphics[scale=0.6]{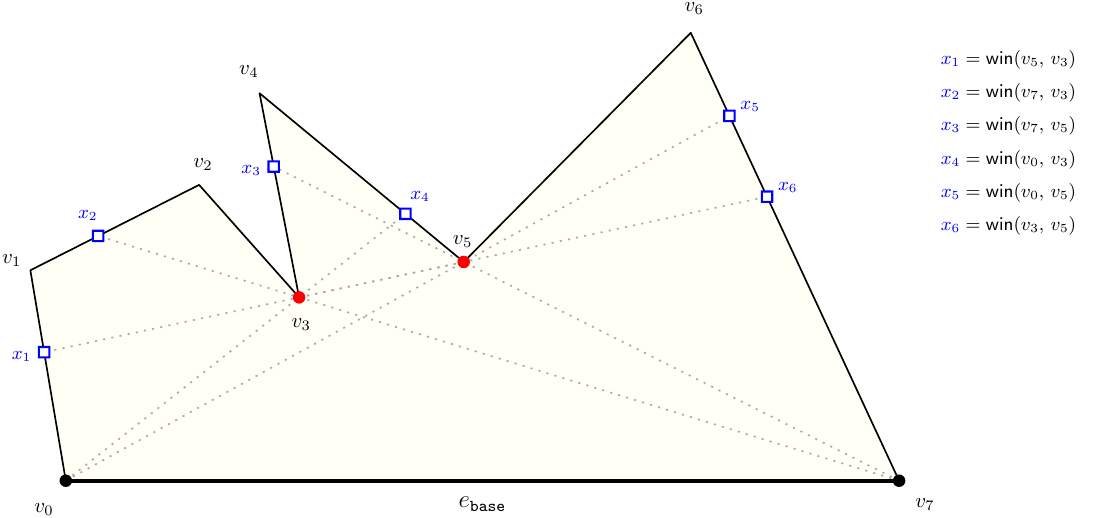}
\vspace{1mm}
\caption{The set $Q = \{v_0, v_3, v_5, v_7\}$ (disks) and the
transition points $x_1, \dots, x_6$ (squares). Each $x_i$ is
$\win(a,b)$ for the dotted segment from $a$ through $b$.}
\label{fig:ws-transition}
\end{subfigure}

\vspace{3ex}

\begin{subfigure}{\textwidth}
\centering
\includegraphics[scale=0.6]{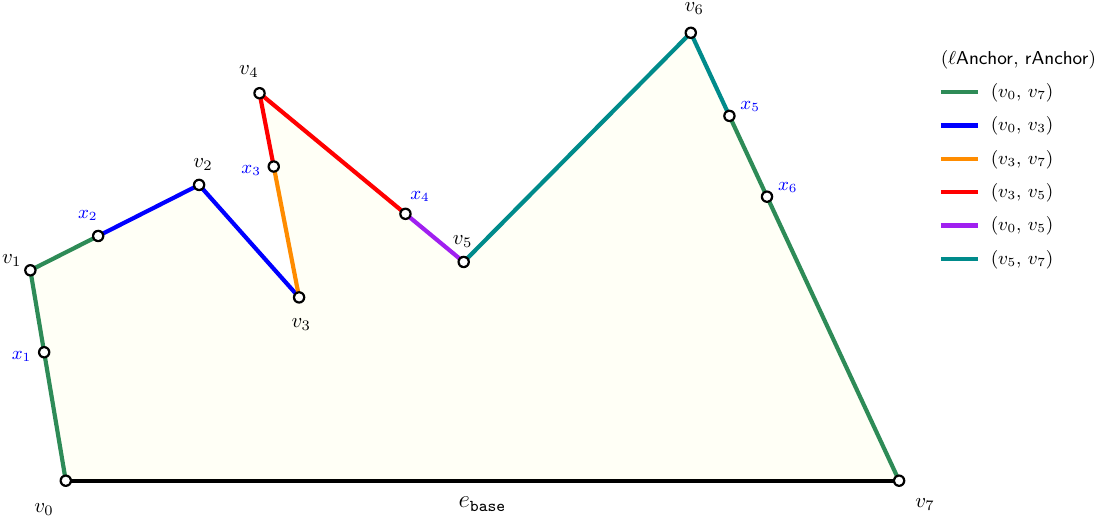}
\vspace{1mm}
\caption{The pieces of the upper chain, coloured by their anchor
pair $(\LA,\RA)$, and the static candidates (circles).}
\label{fig:ws-static}
\end{subfigure}
\vspace{2mm}
\caption{Transition points (\cref{F-def:ws-transition}) and static
candidates (\cref{F-def:ws-static}) in a weak visibility polygon with
base $\eb = v_0v_7$ and reflex vertices $v_3$ and $v_5$. The upper
chain is not monotone. The pairs whose second point is $v_0$ or
$v_7$ contribute the base endpoints themselves to~$\tp$. The anchors
are constant on every piece, but they need not change at a transition
point, as the points $x_1$ and $x_6$ show. Every candidate other than
$v_0$ and $v_7$ ends two pieces and therefore appears twice
in~$\cand$.}
\label{fig:ws-transition-static}
\end{figure}

\begin{definition}[Static Candidates]
\label{F-def:ws-static}
The set $\cand$ of \emph{static candidates} consists of the endpoints of all
pieces (\cref{F-lem:ws-pieces}). Each static candidate carries the two anchors of its
piece and the limits of $\ell$ and $r$ along that piece. A boundary
point that ends two pieces appears twice in $\cand$, once for each
piece.
\end{definition}

The set $\cand$ has at most $2(n + |\tp|)$ members. The window
endpoints are not static. They depend on the base point that is
fired, and the algorithm creates them as it runs.

\subsection{Firing}
\label{F-subsec:ws-firing}

To \emph{fire} a base point $g$ means to return $\win(g, \mu)$ for
every reflex vertex $\mu$ that supports a window of $\vis(g)$.

\begin{observation}
\label{F-obs:fire}
Firing a base point returns at most $\rho$ points, in clockwise boundary order, and it runs in $\OO(n + \rho \log n)$ time. The visibility polygon of a base point is computed in $\OO(n)$ time
\cite{JoeSimpson87}, and each returned point costs $\OO(\log n)$ to
locate its base interval by the anchor queries of \cref{F-sec:prelims}.
\end{observation}

The next lemma describes the base interval of a returned point.

\begin{lemma}
\label{F-lem:ws-fire}
Let $g$ be a base point, let $\mu$ be a reflex vertex that supports a
window of $\vis(g)$, and let $t = \win(g, \mu)$. Then $g$ is an
endpoint of $\I(t)$, and $\mu$ is the anchor of $t$ at that endpoint.
Either $\ell(t) = g$ and $\LA(t) = \mu$, or $r(t) = g$ and
$\RA(t) = \mu$. In the first case, $r(t) > g$.
\end{lemma}

\begin{proof}
The segment from $g$ to $t$ lies in $\wv$ and passes through the
reflex vertex $\mu$. So $t$ sees $g$, and the line of sight grazes
$\mu$. Base points on one side of $g$ are hidden from $t$ by $\mu$.
Hence $g$ is an endpoint of $\I(t)$, and $\mu$ is the reflex vertex on
the line of sight to that endpoint. This gives the two cases. In the
first case, $r(t) > \ell(t) = g$ by general position.
\end{proof}

Chains grow to the right in a second sense. Their right anchors move
clockwise.

\begin{lemma}
\label{F-lem:ws-anchor}
Let $x_1 \lo x_2 \lo \dots \lo x_m$ be a chain. Then
$\RA(x_1) \prec \RA(x_2) \prec \dots \prec \RA(x_m)$. Consequently
$m \le \rho + 1$.
\end{lemma}

\begin{proof}
For consecutive points, $x_j \lo x_{j+1}$ gives
$\RA(x_j) \prec \LA(x_{j+1})$, and every point satisfies
$\LA(x_{j+1}) \preceq \RA(x_{j+1})$. So the right anchors strictly
increase. A right anchor is a reflex vertex or the endpoint $v_n$.
There are at most $\rho + 1$ such points.
\end{proof}

By \cref{F-lem:ws-chain} and \cref{F-lem:ws-anchor}, no witness set has
more than $\rho + 1$ members.

\subsection{The State Table}
\label{F-subsec:ws-table}

Let $\anc$ be the set of possible right anchors. It consists of the
reflex vertices and of $v_n$, so $|\anc| \le \rho + 1$. Let
$\level = \rho + 1$. A \emph{state} is a pair $(i, \mu)$ with
$1 \le i \le \level$ and $\mu \in \anc$.

\begin{definition}[Table Entries]
\label{F-def:ws-table}
For a state $(i, \mu)$, let $\rst(i, \mu)$ be the infimum of
$r(x_i)$ over all chains $x_1 \lo \dots \lo x_i$ of points of $\wv$
with $\RA(x_i) = \mu$. When no such chain exists, we set
$\rst(i, \mu) = +\infty$.
\end{definition}

The largest index $i$ with a finite entry in row $i$ is the size of a
maximum witness set. This follows from \cref{F-lem:ws-chain}. We show
that the table can be filled row by row from finitely many candidates.
Row $i + 1$ is determined by row $i$.

For a state $(i, \mu)$ with finite entry and for $\mu' \in \anc$, let
$S(i, \mu, \mu')$ be the set of all points $x \in \wv$ that satisfy
\[
  \LA(x) \succ \mu, \qquad
  \ell(x) > \rst(i, \mu), \qquad
  \RA(x) = \mu'.
\]
When $\rst(i, \mu) = +\infty$ the set $S(i, \mu, \mu')$ is empty.

\begin{lemma}
\label{F-lem:ws-recursion}
For every state $(i + 1, \mu')$,
\[
  \rst(i + 1, \mu') \;=\;
  \min_{\mu \in \anc}\;
  \inf\bigl\{\, r(x) \bigm| x \in S(i, \mu, \mu') \,\bigr\}.
\]
\end{lemma}

\begin{proof}
Write $S$ for the right-hand side. Let $x_1 \lo \dots \lo x_{i+1}$ be
a chain with $\RA(x_{i+1}) = \mu'$, and let $\mu = \RA(x_i)$. Then
$\LA(x_{i+1}) \succ \mu$ and $\ell(x_{i+1}) > r(x_i) \ge
\rst(i, \mu)$. So $x_{i+1}$ belongs to $S(i, \mu, \mu')$, and
$r(x_{i+1}) \ge S$. Hence $\rst(i + 1, \mu') \ge S$.

Conversely, let $\mu \in \anc$ and let $x$ belong to $S(i, \mu, \mu')$.
Since $\ell(x) > \rst(i, \mu)$, there is a chain $x_1 \lo \dots \lo
x_i$ with $\RA(x_i) = \mu$ and $r(x_i) < \ell(x)$. Also
$\RA(x_i) = \mu \prec \LA(x)$. So $x_i \lo x$, and the chain extends
by $x$. This gives $\rst(i + 1, \mu') \le r(x)$. Taking the infimum
over $x$ and the minimum over $\mu$ gives $\rst(i + 1, \mu') \le S$.
\end{proof}

The infimum in the recursion ranges over all points of $\wv$. The
next lemma shows that finitely many candidates suffice. For a state
$(i, \mu)$ with finite entry we write $g = \rst(i, \mu)$, and we call
a point $x$ a \emph{successor} of the state when $\LA(x) \succ \mu$
and $\ell(x) > g$. The set $S(i, \mu, \mu')$ consists of the
successors with right anchor $\mu'$.

\begin{lemma}
\label{F-lem:ws-update}
Let $(i, \mu)$ be a state with finite entry $g$, and let
$\mu' \in \anc$. The infimum of $r(x)$ over all successors $x$ with
$\RA(x) = \mu'$ equals the minimum of the following two quantities.
\begin{enumerate}
\item The minimum of $r(t)$ over the points $t = \win(g, \nu)$
  returned by firing $g$ that satisfy $\ell(t) = g$, $\nu \succ \mu$,
  and $\RA(t) = \mu'$.
\item The minimum of the limit of $r$ over the static candidates
  $c \in \cand$ whose piece has anchors $\LA \succ \mu$ and
  $\RA = \mu'$ and whose limit of $\ell$ exceeds $g$.
\end{enumerate}
Each minimum is $+\infty$ when its set is empty.
\end{lemma}

\begin{proof}
Both quantities are at least the infimum. A returned point $t$ with
$\ell(t) = g$ is the limit of successors along its piece, and its value
$r(t)$ is the limit of their values by \cref{F-lem:ws-pieces}. A static
candidate is a limit of successors along its piece in the same way.
So the infimum is at most each of the quantities.

Now let $x$ be a successor with $\RA(x) = \mu'$. We find a point in
one of the two sets with value at most $r(x)$. Let $y = \RUC(x)$. By
\cref{F-lem:ws-push}, the point $y$ is a successor with $\RA(y) = \mu'$ and $r(y) = r(x)$. So we may assume that $x$ lies on the upper chain.
If $x$ is a static candidate, we are done. Otherwise, $x$ lies on a
piece. Move $x$ along the piece in the direction that decreases $r$.
By \cref{F-lem:ws-pieces} the value $\ell$ decreases as well, and the
anchors do not change. The motion stops at the first of two events.
Either $x$ reaches an endpoint of the piece, which is a static
candidate with the required anchors and with limit of $\ell$ at least
$g$, or $\ell(x)$ reaches $g$. In the second event, $x$ lies on the
ray from $g$ through $\LA(x)$, and $x$ sees $g$ along this ray. So
$x = \win(g, \LA(x))$, with $\LA(x) \succ \mu$ and $\RA(x) = \mu'$.
In both events, the value of $r$ did not increase.
\end{proof}

We illustrate the entire construction on a single polygon.
\Cref{F-fig:instance} shows a weak visibility polygon with six pockets,
the visibility cone of the apex of each pocket, and the trapezoid that
\cref{F-def:trapezoid} assigns to each apex; the apices are the points
$w_1, \dots, w_6$ are used in the rest of this example.

\begin{figure}[htbp]
\centering
\captionsetup[subfigure]{skip=1pt}

\begin{subfigure}{\textwidth}
\centering
\begin{tikzpicture}[x=0.44cm,y=0.50cm]
  \path (-1.4,-2.0) rectangle (25.4,12.4);
  \fill[black!4] \WVpath;
  \cone[blue!60]   {0.6667}{10.6667}{1.2}{2.8}
  \cone[red!55]    {2.8}{5.6}{10.5}{14}
  \cone[green!65]  {11.5}{9}{7}{8.8}
  \cone[orange!85] {16.2}{5.6}{5}{8.5}
  \cone[violet!60] {18.6667}{10.6667}{16}{17.6}
  \cone[cyan!65]   {23}{8}{19}{21}
  \draw[line width=.9pt] \WVpath;
  \foreach \x/\y/\lb/\pos in {0.6667/10.6667/1/{above},2.8/5.6/2/{above right},
      11.5/9/3/{above right},16.2/5.6/4/{above right},
      18.6667/10.6667/5/{above},23/8/6/{right}}
     {\fill (\x,\y) circle (2.3pt);
      \node[\pos=1.5pt,scale=.82] at (\x,\y) {$w_{\lb}$};}
  \vlab{1}{2} \vlab[8pt]{2}{4} \vlab{5}{5} \vlab[8pt]{6}{7}
  \vlab{9}{8} \vlab[8pt]{10}{10} \vlab{13}{11} \vlab[8pt]{14}{13}
  \vlab{17}{14} \vlab[8pt]{18}{16} \vlab{21}{17} \vlab[8pt]{22}{19}
  \fill (0,4) circle (1.5pt); \fill (24,4) circle (1.5pt);
  \fill (0,0) circle (2pt);
  \node[below left=0pt,scale=.72] at (0,0) {$v_0$};
  \fill (24,0) circle (2pt);
  \node[below right=0pt,scale=.72] at (24,0) {$v_{21}$};
  \node[below=2.5pt,scale=.85] at (12.5,0) {$\eb$};
\end{tikzpicture}

\caption{The polygon $\wv$ has $n=21$ and twelve reflex vertices, and it is
weakly visible from $\eb=v_{21}v_0$.
Each witness $w_i$ sits at the apex of a pocket. The shaded region is
$\vis(w_i)$, whose intersection with the base is the interval $\I(w_i)$. Its two bounding rays pass through the anchors
$\LA(w_i)$ and $\RA(w_i)$. The pockets of $w_2$ and $w_4$ lean in opposite directions, so these two witnesses look across each other, and their visibility polygons meet.}
\label{F-fig:wvp-instance}
\end{subfigure}

\vspace{-2.2ex}

\begin{subfigure}{\textwidth}
\centering
\begin{tikzpicture}[x=0.44cm,y=0.70cm]
  \trap[blue!60]  {\tA}{\tB}{1.2}{2.8}
  \trap[red!55]   {\tC}{\tD}{10.5}{14}
  \trap[orange!85]{\tG}{\tH}{5}{8.5}
  \trap[green!65] {\tE}{\tF}{7}{8.8}
  \trap[violet!60]{\tI}{\tJ}{16}{17.6}
  \trap[cyan!65]  {\tK}{\tL}{19}{21}
  \rails \topticks
  \bt{1.2}{} \bt{2.8}{} \bt{5}{} \bt{7}{} \bt{8.5}{} \bt{8.8}{}
  \bt{10.5}{} \bt{14}{} \bt{16}{} \bt{17.6}{} \bt{19}{} \bt{21}{}
  \node[scale=.72] at (2.3,2) {$T(w_1)$};
  \node[scale=.72] at (6.6,3.1) {$T(w_2)$};
  \node[scale=.72,fill=white,fill opacity=.75,text opacity=1,inner sep=.6pt]
       at (9.1,1.9) {$T(w_3)$};
  \node[scale=.72] at (13.2,1.0) {$T(w_4)$};
  \node[scale=.72] at (17.0,2) {$T(w_5)$};
  \node[scale=.72] at (20.3,2) {$T(w_6)$};
\end{tikzpicture}

\caption{The trapezoid $T(w_i)$ has upper side
$\seg{\tau(\LA(w_i))\,\tau(\RA(w_i))}$ on $L_t$ and lower side
$\seg{\beta(\ell(w_i))\,\beta(r(w_i))}$ on $L_b$. By
\cref{F-thm:trapezoid-graph}, two witnesses conflict if and only if their trapezoids meet. In particular, $T(w_2)$ and $T(w_4)$ cross because the two
orders disagree: on $L_t$ the trapezoid $T(w_2)$ lies to the left, whereas on $L_b$ it lies to the right. The vertices of the upper chain that bound the pockets are drawn at a common height for clarity; an arbitrarily small
perturbation restores the general position assumed in
\cref{F-sec:prelims} and changes none of the values shown.}
\label{F-fig:trap-instance}
\end{subfigure}

\caption{A weak visibility polygon and its trapezoid representation.}
\label{F-fig:instance}
\end{figure}
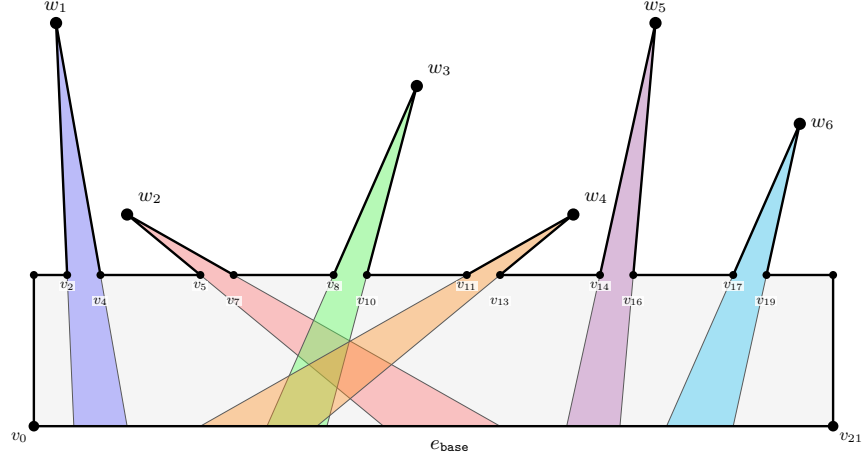
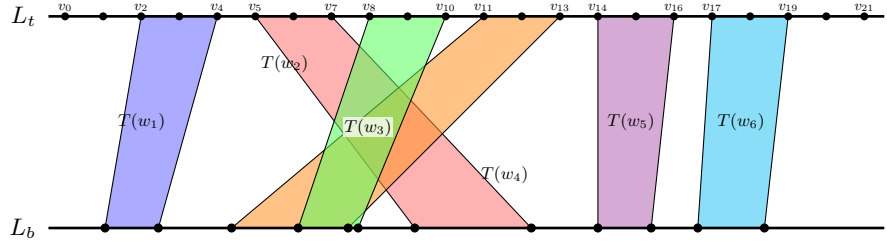


\begin{figure}[htbp]\centering
\captionsetup[subfigure]{skip=1pt}
\begin{subfigure}{\textwidth}\centering
\begin{tikzpicture}[x=0.44cm,y=0.44cm]
  \trap[blue!60]{\tA}{\tB}{1.2}{2.8}
  \traptd[black!45]{\tE}{\tF}{6.6}{9.2}
  \trap[green!65]{\tE}{\tF}{7}{8.8}
  \rails \topticks
  \bt{1.2}{$\ell(w_1)$} \bt[8pt]{2.8}{$r(w_1)$} \bt{7}{$\ell(w_3)$}
  \bt[8pt]{6.6}{$\ell(w_3')$} \bt{8.8}{$r(w_3)$} \bt[8pt]{9.2}{$r(w_3')$}
  \node[scale=.72] at (2.3,2) {$T(w_1)$};
  \node[scale=.72,fill=white,fill opacity=.8,text opacity=1,inner sep=.8pt]
       at (8.8,1.8) {$T(w_3)$};
  \draw[black!65,line width=.4pt] (13.4,1.5) -- (10.1,0.8);
  \node[scale=.72,right,black!65] at (13.5,1.55) {$T(w_3')$};
  \frontier
  \node[red!75!black,scale=.76,above=10pt] at (\tF,\Hgt)
        {$g=\rst(2,v_{10})=r(w_3)$};
\end{tikzpicture}
\caption{The value $g=\rst(2,v_{10})=r(w_3)$, attained by $w_1 \lo w_3$,
defines the dashed frontier from $\beta(g)$ to $\tau(v_{10})$, here the
right side of $T(w_3)$; the competing $w_3'$ shares the anchors of $w_3$
but has a larger $r$, so \cref{F-lem:ws-dom} discards it.}
\label{F-fig:rec-state}
\end{subfigure}
\vspace{-1.2ex}
\begin{subfigure}{\textwidth}\centering
\begin{tikzpicture}[x=0.44cm,y=0.44cm]
  \rightzone
  \trap[violet!60]{\tI}{\tJ}{16}{17.6}
  \trap[cyan!65]{\tK}{\tL}{19}{21}
  \rails \topticks
  \bt{8.8}{$g$} \bt{16}{$\ell(w_5)$} \bt[8pt]{17.6}{$r(w_5)$}
  \bt{19}{$\ell(w_6)$} \bt[8pt]{21}{$r(w_6)$}
  \node[scale=.72] at (17.0,2) {$T(w_5)$};
  \node[scale=.72] at (20.3,2) {$T(w_6)$};
  \frontier
\end{tikzpicture}
\caption{$T(w_5)$ and $T(w_6)$ lies strictly right of the frontier on both
lines, so $w_5$ and $w_6$ are successors and contribute $r(w_5)$ and
$r(w_6)$ to row~$3$.}
\label{F-fig:rec-accept}
\end{subfigure}
\vspace{-1.2ex}
\begin{subfigure}{\textwidth}\centering
\begin{tikzpicture}[x=0.44cm,y=0.44cm]
  \rightzone
  \traptd[red!60]{\tC}{\tD}{10.5}{14}
  \traptd[orange!90]{\tG}{\tH}{5}{8.5}
  \rails \topticks
  \bt{5}{$\ell(w_4)$} \bt[8pt]{8.5}{$r(w_4)$} \bt{8.8}{$g$}
  \bt[8pt]{10.5}{$\ell(w_2)$} \bt{14}{$r(w_2)$}
  \node[scale=.72,fill=white,fill opacity=.8,text opacity=1,inner sep=.8pt]
       at (6.4,3.25) {$T(w_2)$};
  \node[scale=.72,fill=white,fill opacity=.8,text opacity=1,inner sep=.8pt]
       at (13.6,0.75) {$T(w_4)$};
  \frontier
  \draw[red!75!black,line width=.35pt] (16.5,3.05) -- (12.0,2.55);
  \node[red!75!black,scale=.70,right] at (16.6,3.05) {$\LA(w_2)\prec v_{10}$};
  \draw[red!75!black,line width=.35pt] (16.5,1.35) -- (10.5,1.05);
  \node[red!75!black,scale=.70,right] at (16.6,1.35) {$\ell(w_4)\le g$};
\end{tikzpicture}
\caption{Here $w_2$ fails the anchor condition and $w_4$ the base condition,
each satisfying the other.}
\label{F-fig:rec-reject}
\end{subfigure}
\vspace{-1.2ex}
\begin{subfigure}{\textwidth}\centering
\begin{tikzpicture}[x=0.44cm,y=0.44cm]
  \trap[blue!60]{\tA}{\tB}{1.2}{2.8}
  \trap[green!65]{\tE}{\tF}{7}{8.8}
  \trap[violet!60]{\tI}{\tJ}{16}{17.6}
  \trap[cyan!65]{\tK}{\tL}{19}{21}
  \rails \topticks
  \bt{1.2}{} \bt{2.8}{} \bt{7}{} \bt{8.8}{} \bt{16}{} \bt{17.6}{}
  \bt{19}{} \bt{21}{}
  \node[scale=.72] at (2.3,2)  {$T(w_1)$};
  \node[scale=.72] at (8.7,2)  {$T(w_3)$};
  \node[scale=.72] at (17.0,2) {$T(w_5)$};
  \node[scale=.72] at (20.3,2) {$T(w_6)$};
\end{tikzpicture}
\caption{Iterating gives $w_1 \lo w_3 \lo w_5 \lo w_6$, a
witness set by \cref{F-lem:ws-chain}.}
\label{F-fig:rec-chain}
\end{subfigure}
\caption{One step of the recursion of \cref{F-lem:ws-recursion} and
\cref{F-lem:ws-update} from the state $(2,v_{10})$, on the polygon of
\cref{F-fig:wvp-instance}.}
\label{F-fig:recursion}
\end{figure}
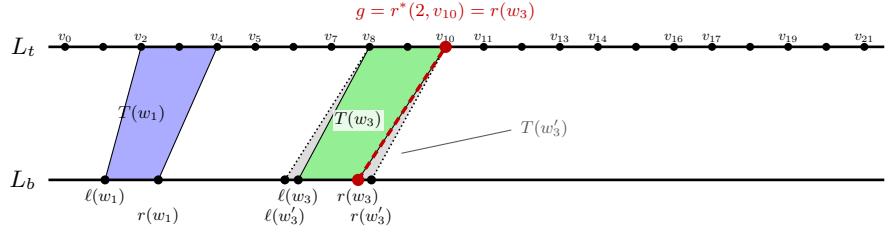
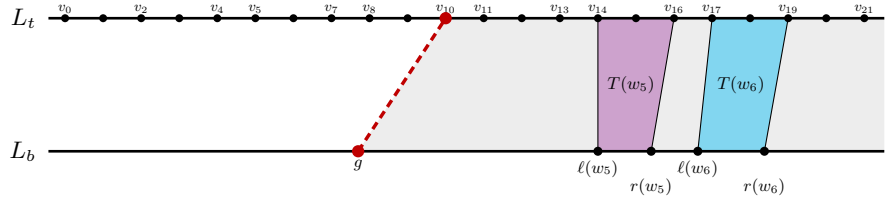
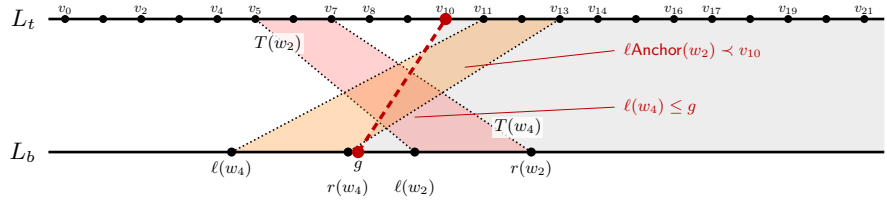
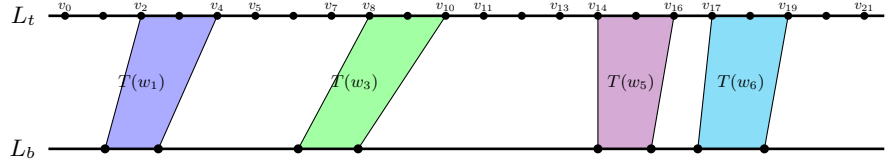

\Cref{F-fig:recursion} carries out one step of the recursion on the instance
of \cref{F-fig:instance}, from the state $(2, v_{10})$.

The algorithm fills the table by \cref{F-lem:ws-recursion} and
\cref{F-lem:ws-update}. It records for each finite entry the candidate
that achieved the minimum and the state it came from; \cref{F-fig:state-table}
shows the completed table for the instance of \cref{F-fig:instance}.

\begin{algorithm}[htbp]
\DontPrintSemicolon
\caption{The state table algorithm.}
\label{F-alg:ws-table}
compute $\tp$, the pieces, and the static candidates $\cand$\;
set every entry $\rst(i, \mu)$ to $+\infty$\;
\ForEach{static candidate $c \in \cand$ with anchors $(\nu, \mu')$}{
  $\rst(1, \mu') \gets \min\{\rst(1, \mu'),\ r(c)\}$\;
}
\For{$i \gets 1$ \KwTo $\level - 1$}{
  \ForEach{$\mu \in \anc$ with $\rst(i, \mu) < +\infty$}{
    $g \gets \rst(i, \mu)$\;
    fire $g$ and compute the anchors of every returned point\;
    \ForEach{returned point $t$ with $\ell(t) = g$ and $\LA(t) \succ \mu$}{
      $\rst(i + 1, \RA(t)) \gets \min\{\rst(i + 1, \RA(t)),\ r(t)\}$\;
    }
    \ForEach{static candidate $c$ with anchors $(\nu, \mu')$,
      $\nu \succ \mu$, and limit of $\ell$ above $g$}{
      $\rst(i + 1, \mu') \gets \min\{\rst(i + 1, \mu'),\ r(c)\}$\;
    }
  }
}
$k \gets$ the largest $i$ with a finite entry in row $i$\;
\Return a witness set of size $k$ built from the recorded candidates
 by \cref{F-lem:ws-output}\;
\end{algorithm}

\begin{figure}[H]\centering
\begin{tikzpicture}[x=1.44cm,y=1.04cm,
   cell/.style={draw,black!60,line width=.4pt,minimum width=1.30cm,
                minimum height=.64cm,inner sep=1pt,anchor=center},
   hit/.style ={cell,fill=green!22,line width=1pt,draw=black},
   hd/.style  ={anchor=center,scale=.95},
   ar/.style  ={-{Stealth[length=4.5pt]},line width=.9pt,rounded corners=2pt},
   lbl/.style ={fill=white,inner sep=1.3pt,scale=.8,text=black}]
 
  \foreach \j/\lab in {1/{$v_4$},2/{$v_7$},3/{$v_{10}$},4/{$v_{13}$},
                       5/{$v_{16}$},6/{$v_{19}$},7/{$v_{21}$}}
     {\node[hd] at (\j,5.62) {\lab};}
  \node[hd,scale=.9] at (-0.06,5.62) {$i \,\backslash\, \mu$};
  \foreach \i in {1,...,5}{\node[hd] at (-0.06,6-\i) {$\i$};}
 
  \node[hit] (A) at (1,5) {$2.8$};    \node[cell] at (2,5) {$14$};
  \node[cell]     at (3,5) {$8.8$};   \node[cell] at (4,5) {$8.5$};
  \node[cell]     at (5,5) {$17.6$};  \node[cell] at (6,5) {$21$};
  \node[cell]     at (7,5) {$24$};
 
  \node[cell]     at (1,4) {$+\infty$};\node[cell] at (2,4) {$14$};
  \node[hit] (B)  at (3,4) {$8.8$};    \node[cell] at (4,4) {$8.5$};
  \node[cell]     at (5,4) {$17.6$};   \node[cell] at (6,4) {$21$};
  \node[cell]     at (7,4) {$+\infty$};
 
  \node[cell]     at (1,3) {$+\infty$};\node[cell] at (2,3) {$+\infty$};
  \node[cell]     at (3,3) {$+\infty$};\node[cell] at (4,3) {$+\infty$};
  \node[hit] (C)  at (5,3) {$17.6$};   \node[cell] at (6,3) {$21$};
  \node[cell]     at (7,3) {$+\infty$};
 
  \node[cell]     at (1,2) {$+\infty$};\node[cell] at (2,2) {$+\infty$};
  \node[cell]     at (3,2) {$+\infty$};\node[cell] at (4,2) {$+\infty$};
  \node[cell]     at (5,2) {$+\infty$};\node[hit] (D) at (6,2) {$21$};
  \node[cell]     at (7,2) {$+\infty$};
 
  \foreach \j in {1,...,7}{\node[cell] at (\j,1) {$+\infty$};}
 
  \draw[ar] (1,4.78) -- (1,4.50) -- (3,4.50) -- (3,4.22);
  \node[lbl] at (2,4.50) {$w_3$};
  \draw[ar] (3,3.78) -- (3,3.50) -- (5,3.50) -- (5,3.22);
  \node[lbl] at (4,3.50) {$w_5$};
  \draw[ar] (5,2.78) -- (5,2.50) -- (6,2.50) -- (6,2.22);
  \node[lbl] at (5.5,2.50) {$w_6$};
\end{tikzpicture}
\caption{The state table for the polygon of
\cref{F-fig:wvp-instance}, filled row by row by Lemmas~\ref{F-lem:ws-recursion} and~\ref{F-lem:ws-update}. A column is shown for each anchor that is the $\RA$ of some point of $\wv$; the remaining six columns of $\anc$ are
identically $+\infty$, as are the rows below the fifth, and both are
omitted. The arrows record, for each shaded entry, the candidate that
achieved the minimum and the state it came from. Row~$5$ is empty, so the
largest index with a finite entry is $k=4$, which by Lemmas~\ref{F-lem:ws-chain} and~\ref{F-lem:ws-anchor} is the size of a maximum witness set; following the records backward, as in \cref{F-lem:ws-output}, returns the chain
$w_1 \lo w_3 \lo w_5 \lo w_6$.}
\label{F-fig:state-table}
\end{figure}

\begin{theorem}[Correctness]
\label{F-thm:ws-correct}
\Cref{F-alg:ws-table} computes every entry $\rst(i, \mu)$ of the table.
The number $k$ that it returns is the size of a maximum witness set of
$\wv$.
\end{theorem}

\begin{proof}
We use induction on the row index. Row $1$ is correct. The entry
$\rst(1, \mu')$ is the infimum of $r$ over all points with right anchor
$\mu'$. By \cref{F-lem:ws-push}, the infimum over $\wv$ equals the
infimum over the upper chain. By \cref{F-lem:ws-pieces} the infimum over
the upper chain is attained in the limit at a piece endpoint, that is,
at a static candidate. Now assume that row $i$ is correct. For each
state of row $i$ with a finite entry, the algorithm performs the update
of \cref{F-lem:ws-update}. By \cref{F-lem:ws-recursion} this yields row
$i + 1$. The claim about $k$ follows from \cref{F-lem:ws-chain},
\cref{F-lem:ws-anchor}, and \cref{F-def:ws-table}.
\end{proof}

The recorded candidates form a chain in a weak sense. Two consecutive
candidates may touch on the base. A small shift removes the contact.

\begin{lemma}
\label{F-lem:ws-output}
Let $\rst(k, \mu) < +\infty$. Following the records backward from
$(k, \mu)$ gives candidates $y_1, \dots, y_k$ with
$\RA(y_k) = \mu$. These candidates can be shifted along their pieces
into points $x_1 \lo x_2 \lo \dots \lo x_k$. The set
$\{x_1, \dots, x_k\}$ is a witness set of size $k$.
\end{lemma}

\begin{proof}
Each record was produced by \cref{F-lem:ws-update}. So for every index
$j < k$, the candidate $y_{j+1}$ is a limit of successors of the state
that contains $y_j$, along a piece $\alpha_{j+1}$. In particular
$\LA(y_{j+1}) \succ \RA(y_j)$ and $\ell(y_{j+1}) \ge r(y_j)$, where
the inequality may be an equality. We process $j = 2, \dots, k$ in
this order. Move $y_j$ along $\alpha_j$ into the direction that
increases $\ell$, by a positive distance small enough that
$\ell(y_j) > r(y_{j-1})$ holds for the shifted point. Such a distance
exists because $y_j$ is a limit of successors along $\alpha_j$. The
shift increases $r(y_j)$ by a small amount, and the next step will
move $y_{j+1}$ past the new value. After the last step, consecutive
points satisfy $\ell(x_{j+1}) > r(x_j)$ and
$\LA(x_{j+1}) \succ \RA(x_j)$, because anchors do not change along a
piece. Hence $x_j \lo x_{j+1}$ for every $j$, and the points form a
chain by transitivity. By \cref{F-lem:ws-chain} they form a witness set.
\end{proof}

\begin{corollary}
\label{F-cor:ws-maxset}
\Cref{F-alg:ws-table} returns a maximum witness set of $\wv$.
\end{corollary}

Backtracking through the records of \cref{F-fig:state-table} returns the four
points drawn in \cref{F-fig:max-witness}, which is a maximum witness set of
that polygon.

\begin{figure}[htbp]\centering
\begin{tikzpicture}[scale=0.45]
  \path (-1.4,-2.0) rectangle (25.4,12.4);
  \fill[black!4] \WVpath;
  \cone[blue!60]   {0.6667}{10.6667}{1.2}{2.8}
  \cone[green!65]  {11.5}{9}{7}{8.8}
  \cone[violet!60] {18.6667}{10.6667}{16}{17.6}
  \cone[cyan!65]   {23}{8}{19}{21}
  \draw[line width=.9pt] \WVpath;
  \foreach \x/\y/\lb/\pos in {2.8/5.6/2/{above right},16.2/5.6/4/{above right}}
     {\draw[black!55,line width=.8pt] (\x-.32,\y-.32) -- (\x+.32,\y+.32);
      \draw[black!55,line width=.8pt] (\x-.32,\y+.32) -- (\x+.32,\y-.32);
      \node[\pos=1pt,scale=.8,black!55] at (\x,\y) {$w_{\lb}$};}
  \foreach \x/\y/\lb/\pos in {0.6667/10.6667/1/{above},11.5/9/3/{above right},
      18.6667/10.6667/5/{above},23/8/6/{right}}
     {\fill (\x,\y) circle (2.5pt);
      \node[\pos=1.5pt,scale=.85] at (\x,\y) {$w_{\lb}$};}
  \foreach \l/\r/\c in {1.2/2.8/{blue!60!black},7/8.8/{green!45!black},
                        16/17.6/{violet!70!black},19/21/{cyan!55!black}}
     {\draw[line width=2.2pt,\c] (\l,0.07) -- (\r,0.07);
      \fill[\c] (\l,0.07) circle (2pt); \fill[\c] (\r,0.07) circle (2pt);}
  \fill (0,0) circle (2pt); \node[below left=0pt,scale=.72] at (0,0) {$v_0$};
  \fill (24,0) circle (2pt);\node[below right=0pt,scale=.72] at (24,0) {$v_{21}$};
  \node[below=2.5pt,scale=.85] at (12.5,0) {$\eb$};
\end{tikzpicture}
\caption{The maximum witness set $\{w_1,w_3,w_5,w_6\}$ returned by
\cref{F-alg:ws-table} on the polygon of \cref{F-fig:wvp-instance}, so that
$\witN(\wv,\wv)=4$. The four visibility polygons are pairwise disjoint, as
\cref{F-def:witness-set} demands, and the four base intervals $\I(w_i)$ drawn in bold are pairwise disjoint and ordered from left to right, in agreement
with \cref{F-obs:vis-implies-interval}. The two witnesses marked with a cross
are excluded: each of them conflicts with a member of the set, and no witness set of this polygon has five members. The optimum is not unique,
since $\{w_1,w_2,w_5,w_6\}$ is a second maximum witness set.}
\label{F-fig:max-witness}
\end{figure}

\subsection{Termination}
\label{F-subsec:ws-termination}

The algorithm processes each state at most once and fires one base
point per state. The number of anchors bounds the number of states. This gives termination without any hypothesis on the
polygon.

\begin{theorem}
\label{F-thm:ws-terminate}
\Cref{F-alg:ws-table} fires at most $(\rho + 1)^2$ base points, creates
at most $\rho (\rho + 1)^2$ window endpoints, and stops.
\end{theorem}

\begin{proof}
A base point is fired once for each state with a finite entry. There
are at most $\level \cdot |\anc| \le (\rho + 1)^2$ states. Each firing
returns at most $\rho$ points by \cref{F-obs:fire}. Every loop of the
algorithm runs over a finite set, so the algorithm stops.
\end{proof}

We close with a remark on why such a bound is needed.
Firing produces points whose right endpoints move to the right by
\cref{F-lem:ws-fire}. Firing a returned point again can produce a
further point with the same two anchors, and so on. This cascade
consists of points that touch each other on the base, and their base
intervals tile a part of $\eb$ from left to right. The cascade can be
infinite. Its points all share one right anchor, and their right
endpoints increase. By \cref{F-lem:ws-dom} the first point of the
cascade can replace every later one. \Cref{F-alg:ws-table} never fires a
point of such a cascade after the first, because it fires only table
entries, and a table entry is the smallest right endpoint of its
anchor. This is the reason why the algorithm is stated in terms of
states rather than in terms of individual candidates.


\section{Running Time}
\label{F-sec:ws-time}

We bound the running time of \cref{F-alg:ws-table}.

\paragraph{Preprocessing.}
We triangulate $\wv$ in $\OO(n)$ time \cite{Chazelle91} and build the
shortest path trees from $v_0$ and from $v_n$ in $\OO(n)$ time
\cite{GHLST87}. These trees answer an anchor query for any point of
$\bd(\wv)$ in $\OO(\log n)$ time. A ray shooting structure with
$\OO(n)$ preprocessing answers a query in $\OO(\log n)$ time
\cite{DBLP:journals/jal/HershbergerS95}. The set $\tp$ has $\OO(\rho^2)$ members, and a single ray-shooting query finds each.
We sort the vertices and the transition points along the upper chain
in $\OO((n + \rho^2) \log n)$ time. This yields the pieces. For each piece, we query the anchors of an interior point once, and we compute
the limits of $\ell$ and $r$ at its two endpoints by two projections.
The preprocessing therefore takes $\OO((n + \rho^2) \log n)$ time.

\paragraph{One state.}
Consider a state $(i, \mu)$ with finite entry $g$. The visibility
polygon of the base point $g$ is computed in $\OO(n)$ time
\cite{JoeSimpson87}. Its windows and their far endpoints are read simultaneously. There are at most $\rho$ returned points, and
the anchors of each cost $\OO(\log n)$. The scan over the static
candidates takes $\OO(n + \rho^2)$ time. So one state costs
$\OO(n + \rho^2 + \rho \log n)$ time, which is $\OO(n + \rho^2)$.

\paragraph{All states.}
By \cref{F-thm:ws-terminate} there are at most $(\rho + 1)^2$ states
with a finite entry. The output step of \cref{F-lem:ws-output} shifts
$k \le \rho + 1$ points, each in $\OO(\log n)$ time. Adding the
preprocessing gives the bound.

\begin{theorem}[Running Time]
\label{F-thm:ws-time}
\Cref{F-alg:ws-table} computes a maximum witness set of a weak
visibility polygon with $n$ vertices and $\rho$ reflex vertices in
\[
  \OO\bigl(\, n \log n + \rho^{2} (n + \rho^{2}) \,\bigr)
\]
time.
\end{theorem}

For polygons with $\rho = \OO(\sqrt{n})$ reflex vertices the bound is
$\OO(n^{2})$. The term $\rho^{2} n$ comes from computing one
visibility polygon per state. Replacing this computation by $\rho$
ray shooting queries per state, and replacing the scan over $\cand$ by
a two-dimensional range minimum structure, lowers the bound to
$\OO(n \log n + \rho^{3} \log^{2} n)$. We do not need this refinement
for the results of this paper.


\section{Conclusion}
\label{F-sec:conclusion}

We gave the first exact polynomial time algorithms for witness sets in weak
visibility polygons, in both the discrete and the continuous model. The discrete algorithm is near-linear and optimal, and it rests on a structural
fact of independent interest: the visibility intersection graph of a weak visibility polygon is a trapezoid graph, and the class of these graphs in turn
properly contains the interval graphs and the permutation graphs.

Monotone mountains are perfect: their guard number equals their witness
number~\cite{DAESCU201922}. Weak visibility polygons are not. There are such
polygons on $n$ vertices that require $\Omega(n)$ guards and whose maximum
witness set is a single point, so the certificate a maximum witness set
provides can be as weak as the trivial one. The witness number of a weak
visibility polygon is therefore worth computing in its own right, and not as a
proxy for the guard number.

Our work leaves three concrete questions.
\begin{enumerate}
  \item Is every trapezoid graph a visibility intersection graph of some weak
  visibility polygon? A positive answer would characterize the conflict
  structure of $\dWSP$ exactly.
  \item Our lower bound for $\dWSP$ is tight when $m = \Theta(n)$. For larger
  $m$ the picture changes, because a witness set has at most $\rho + 1$
  members, and most of the input is therefore discarded. What is the exact
  complexity in that regime?
  \item Can the running time for $\conWSP$ be improved? Replacing each
  visibility polygon computation by $\rho$ ray shooting queries, and the scan
  over the static candidates by a two-dimensional range minimum structure,
  appears to give $\OO(n \log n + \rho^{3} \log^{2} n)$.
\end{enumerate}

\paragraph*{Declaration on the use of AI.}
The authors used Claude (Anthropic) to assist with drafting and revising portions of the manuscript text. All mathematical results, proofs, and technical content are the authors' own.

\bibliographystyle{splncs04}
\bibliography{udvas}

\end{document}